\documentclass{amsart}
\usepackage[letterpaper,margin=1.25in]{geometry}
\usepackage{lmodern}
\usepackage[utf8]{inputenc}
\usepackage[dvipsnames,svgnames,x11names]{xcolor}
\usepackage{mathtools,microtype}
\usepackage{graphicx}
\usepackage{xparse}

\usepackage{thmtools}

\usepackage{mathrsfs}

\usepackage{setspace}
\usepackage[
    backref=page,
    colorlinks,
    citecolor=Mahogany,
    linkcolor=Mahogany,
    urlcolor=Mahogany,
    filecolor=Mahogany,
    hypertexnames=false
]{hyperref}

\usepackage{cancel}

\usepackage{booktabs}

\usepackage{tikz,tikz-cd}
\usetikzlibrary{matrix,calc,positioning,arrows,decorations.pathreplacing,patterns,knots}

\usepackage[capitalise]{cleveref}
\usepackage{quiver}

\newcommand{\bC}{\mathbb{C}}

\newcommand{\bR}{\mathbb{R}}
\newcommand{\bT}{\mathbb{T}}

\newcommand{\id}{\mathrm{id}}\newcommand{\C}{\mathbb{C}}
\newcommand{\Z}{\mathbb{Z}}
\newcommand{\R}{\mathbb{R}}

\renewcommand{\O}{\mathrm{O}}
\newcommand{\Spin}{\mathrm{Spin}}

\newcommand{\FtwoI}{\mathrm{F2I}}
\newcommand{\ABS}{\mathrm{ABS}}

\definecolor{lavenderblue}{rgb}{0.8, 0.8, 1.0}
\definecolor{lightmauve}{rgb}{0.86, 0.82, 1.0}
\definecolor{brightlavender}{rgb}{0.75, 0.58, 0.89}

\numberwithin{equation}{section}

\usepackage{wrapfig,caption}
\newtheorem{theorem}[equation]{Theorem}
\newtheorem{lemma}[equation]{Lemma}

\newtheorem{proposition}[equation]{Proposition}
\newtheorem{corollary}[equation]{Corollary}
\newtheorem*{corollary*}{Corollary}

\newtheorem{ansatz}[equation]{Ansatz}
\newtheorem*{ans*}{Ansatz}

\theoremstyle{definition}

\newtheorem{definition}[equation]{Definition}

\newtheorem{example}[equation]{Example}

\crefname{theorem}{Theorem}{Theorems}
\crefname{lemma}{Lemma}{Lemmas}
\crefname{definition}{Definition}{Definitions}
\crefname{proposition}{Proposition}{Propositions}
\crefname{corollary}{Corollary}{Corollaries}
\crefname{ansatz}{Ansatz}{Ansatzes}

\theoremstyle{remark}

\newtheorem{remark}[equation]{Remark}

\newtheorem{litnote}[equation]{Literature Note}

\crefname{example}{Example}{Examples}
\crefname{litnote}{Literature Note}{Literature Notes}
\crefname{remark}{Remark}{Remarks}

\usepackage{xargs}                      

\definecolor{todocolor}{HTML}{FF0000}
\usepackage[colorinlistoftodos,prependcaption,textsize=tiny,linecolor=todocolor,backgroundcolor=todocolor!20,bordercolor=todocolor]{todonotes}
\DeclareDocumentCommand{\shortexact}{s O{} O{} mmmm}{
\IfBooleanTF{#1}{ 
\begin{tikzcd}[ampersand replacement=\&]
        1 \& {#4}
        \&  {#5}
        \& {#6}
        \& 1#7
        \arrow[from=1-1, to=1-2]
        \arrow["#2", from=1-2, to=1-3]
        \arrow["#3", from=1-3, to=1-4]
        \arrow[from=1-4, to=1-5]
\end{tikzcd}
}{ 
\begin{tikzcd}[ampersand replacement=\&]
        0 \& {#4}
        \&  {#5}
        \& {#6}
        \& 0#7
        \arrow[from=1-1, to=1-2]
        \arrow["#2", from=1-2, to=1-3]
        \arrow["#3", from=1-3, to=1-4]
        \arrow[from=1-4, to=1-5]
\end{tikzcd}
}}

\definecolor{natcolor}{HTML}{AAAADD}

\definecolor{luukcolor}{HTML}{1d3dc4}

\definecolor{camcolor}{HTML}{43b3ae}

\definecolor{aruncolor}{HTML}{800000}

\definecolor{omarcolor}{HTML}{FF6E00}

\definecolor{danielcolor}{HTML}{0A5C36}

\newcommand{\tikzs}{\begin{center}\begin{tikzcd}}
\newcommand{\tikze}{\end{tikzcd}\end{center}}

\newcommand{\wh}[1]{\widehat{#1}}

\newcommand{\pt}{\mathrm{pt}}
\definecolor{softred}{rgb}{0.92, 0, 0.17}

\usepackage{xspace}
\usepackage{slashed}

\newcommand{\spinc}{\text{spin$^c$}\xspace}

\newcommand{\pinp}{pin\textsuperscript{$+$}\xspace}
\newcommand{\pincp}{pin\textsuperscript{$\tilde c+$}\xspace}

\newcommand{\KO}{\mathit{KO}}
\newcommand{\KR}{\mathit{KR}}

\newcommand{\ol}{\overline}
\newcommand{\Hom}{\operatorname{Hom}}

\renewcommand{\O}{\mathrm{O}}
\newcommand{\SO}{\mathrm{SO}}
\newcommand{\U}{\mathrm{U}}
\newcommand{\SU}{\mathrm{SU}}
\newcommand{\Pin}{\mathrm{Pin}}

\newcommand{\RP}{\mathbb{RP}}
\newcommand{\CP}{\mathbb{CP}}

\newcommand{\MTSpin}{MT\mathrm{Spin}}

\newcommand{\Cl}{\mathit{C\hspace{-0.5pt}\ell}}

\newcommand{\Ext}{\mathrm{Ext}}
\newcommand{\sm}{\mathrm{sm}}
\newcommand{\Snb}{S_{\mathit{nb}}}
\renewcommand{\Im}{\mathrm{Im}}

\newcommand{\SH}{\mathit{SH}}

\newcommand{\term}{\emph}
\DeclarePairedDelimiter{\set}{\{}{\}}

\newcommand{\coker}{\mathrm{coker}\,}

\newcommand{\pinm}{pin\textsuperscript{$-$}\xspace}
\newcommand{\pincm}{pin\textsuperscript{$\tilde c-$}\xspace}
\newcommand{\spinh}{spin\textsuperscript{$h$}\xspace}
\newcommand{\pinhm}{pin\textsuperscript{$h-$}\xspace}
\newcommand{\pinhp}{pin\textsuperscript{$h+$}\xspace}

\definecolor{darkcerulean}{rgb}{0.03, 0.27, 0.49}
\definecolor{ceruleanblue}{rgb}{0.16, 0.32, 0.75}
\definecolor{cobalt}{rgb}{0.0, 0.28, 0.67}
\definecolor{denim}{rgb}{0.08, 0.38, 0.74}

\newcommand{\nocontentsline}[3]{}
\let\origcontentsline\addcontentsline
\newcommand\stoptoc{\let\addcontentsline\nocontentsline}
\newcommand\resumetoc{\let\addcontentsline\origcontentsline}

\title{Weak topological phases in the presence of interactions}
\date{\today}
\author[Antolín Camarena]{Omar Antolín Camarena}
\address{Instituto de Matemáticas, UNAM, Circuito exterior, Ciudad Universitaria, 04510 Ciudad de México, CDMEX, Mexico}
\email{\href{mailto:omar@matem.unam.mx}{omar@matem.unam.mx}}
\urladdr{\href{https://www.matem.unam.mx/~omar/}{https://www.matem.unam.mx/~omar/}}

\author[Debray]{Arun Debray}
\address{Department of Mathematics, The University of Kentucky, 719 Patterson Office Tower,
Lexington, KY 40506, USA}
\email{\href{mailto:a.debray@uky.edu}{a.debray@uky.edu}}
\urladdr{\href{https://adebray.github.io/}{https://adebray.github.io/}}

\author[Krulewski]{Cameron Krulewski}
\address{Department of Mathematics \& Statistics,
Dalhousie University,
6316 Coburg Road,
PO Box 15000,
Halifax, NS, B3H 4R2, Canada}
\email{\href{mailto:ckrulewski@dal.ca}{ckrulewski@dal.ca}}
\urladdr{\href{https://cakrulewski.github.io/}{https://cakrulewski.github.io/}}

\author[Pacheco-Tallaj]{Natalia Pacheco-Tallaj}
\address{Massachusetts Institute of Technology,
Department of Mathematics,
Simons Building (Building 2)
77 Massachusetts Avenue,
Cambridge, MA 02139, USA}
\email{\href{mailto:nataliap@mit.edu}{nataliap@mit.edu}}
\urladdr{\href{https://natipt.github.io/}{https://natipt.github.io/}}

\author[Sheinbaum]{Daniel Sheinbaum}
\address{Division of Applied Physics, CICESE 22860, Ensenada, BC, Mexico.}
\email{\href{mailto:daniels@cicese.mx}{daniels@cicese.mx}}
\urladdr{\href{https://d-shein.github.io/}{https://d-shein.github.io/}}

\author[Stehouwer]{Luuk Stehouwer}
\address{Department of Mathematical Sciences,
Durham University,
Stockton Rd, Durham DH1 3LE, United Kingdom}
\email{\href{luukstehouwer@gmail.com}{luukstehouwer@gmail.com}}
\urladdr{\href{https://sites.google.com/view/luuk-stehouwer}{https://sites.google.com/view/luuk-stehouwer}}

\begin{document}

\begin{abstract}
We study weak symmetry-protected topological phases (SPTs) in the presence of short-range interactions. By comparing homotopical free and interacting classifications of these SPTs, we 
predict
their stability under interactions as well as 
identify potential
intrinsically-interacting phases.
We mathematically compute the groups of weak phases in dimensions zero through three for all tenfold-way symmetry types using homotopy theory; specifically, we use
Atiyah's Real {\itshape KR}-theory and the low-energy invertible field theory ansatz of Freed--Hopkins for the free and interacting cases, resp.
Our computational techniques involve T-duality, {which relates {\itshape K}-theory of the spatial torus with {\itshape K}-theory of the Brillouin torus}, and {a binomial formula for computing generalized cohomology of a torus}.
Our results carry potential implications for theoretical and experimental studies of 
weak phases.

\end{abstract}

\maketitle

\tableofcontents

\setcounter{tocdepth}{3}
\makeatletter
\providecommand\@dotsep{5}
\makeatother

\setcounter{section}{-1}
\section{Introduction}

From at least the late 1980s,
it has been clear that homotopical techniques are useful for classifying what have come to be called symmetry-protected topological phases of matter (SPTs). For free SPTs, there are successful classification frameworks using $K$-theory groups, as most famously put forth by Kitaev in his periodic table of topological insulators and superconductors \cite{kitaev_periodic_2009}. 
However, as first observed by Fidkowski-Kitaev~\cite{fidkowski_effects_2010, FK11} and confirmed by Turner-Pollman-Berg~\cite{turner_topological_2011}, once one allows a Hamiltonian model to have short-range interaction terms, the classification will in general change, going beyond $K$-theory.
Subsequently, homotopical classifications for interacting SPTs (under appropriate assumptions) have been developed that leverage the mathematical framework of topological quantum field theory (TQFT), as we will expand upon later.
{These classifications rely on the ansatz that the equivalence class of an SPT depends only on its low-energy effective field theory, and that this theory is an invertible field theory. (See e.g.\ \cite{Fre19,freed_reflection_2021}.)}
{Under this ansatz, homotopy theory has been widely useful for computing}
free and interacting classifications of SPTs, as in, for example, \cite{freed_twisted_2013, freed_k-theory_2016, freed_reflection_2021, Cam17, KT17, shiozaki_topological_2017, BC18, DNG18b, WG18, Xiong_SPT, Gaiotto_SPT, freed_invertible_2019, debray_invertible_2021, BCHM22, de_nittis_cohomology_2022, shiozaki_AHSS_2022, DNG23, shiozaki_generalized_2023, beaudry2023homotopical, lee_connection_2024, lee_crystalline_2024, de_nittis_topological_2025}.

It is less well-known that homotopical methods can actually offer even more information: 
they can{---again, admitting the low-energy IFT ansatz---}precisely compute the interacting phase of a given free fermion Hamiltonian.
In this paper, we use the framework of Freed-Hopkins \cite{freed_reflection_2021} to study a group homomorphism between free and interacting classifications for a subclass of SPTs.
This homomorphism, called a \textit{free-to-interacting map}, mathematically 
{predicts}
(via its kernel) the free phases that are unstable to short-range interactions and (via its cokernel) the interacting phases that do not arise from free phases: so-called interaction-enabled phases.
Roughly speaking, the homomorphism takes a free Hamiltonian and then ``turns on interactions,'' considering the model in a  space of interacting Hamiltonians of the same symmetry type.
This homotopical approach complements and provides a mathematical proof for physical calculations that are usually performed on a case-by-case basis using specific lattice models. For example, see \cite{yao_interaction_2013,cho_topological_2015,shiozaki_topology_2016,lu_classification_2017,song_interaction_2017,thorngren_gauging_2018,liu_shift_2019,shiozaki_generalized_2023, lee_connection_2024, lee_crystalline_2024}.
{Conversely, comparing homotopical approaches with physical computations provides further evidence and use cases for the low-energy IFT ansatz.}

Freed-Hopkins studied and computed homotopical free-to-interacting maps in low degrees for all symmetry types in the famous tenfold way of free fermions \cite[\S 9.2]{freed_reflection_2021}.
In the tenfold way, 
by specifying whether a phase has time reversal symmetry, charge-conjugation symmetry, and/or chiral symmetry, one obtains one out of ten Altland-Zirnbauer classes \cite{AltlandZirnbauer_1997}.
However, within each tenfold way class, there are different kinds of phases, and Freed-Hopkins considered only the so-called \textit{strong} phases: those which lack any spatial symmetries and are instead protected solely by internal symmetries.
In our work, we extend Freed-Hopkins' free-to-interacting maps to study \emph{weak} SPTs, which are phases protected by a lattice translation symmetry.

Concretely, weak SPTs can be formed out of layers of lower-dimensional strong SPTs, and so invariants of these phases are often assembled from invariants of the constituent strong phases ~\cite{fu_topological_2007-2, claes_disorder_2020, Liu_2016_WTI_Weyl_semimetals}.
In addition to providing a testing ground for homotopical methods, weak phases are of theoretical and experimental interest.
For example, their construction by stacking {often results in anisotropic gapless edge modes} \cite{claes_disorder_2020, Liu_2016_WTI_Weyl_semimetals, Yan_2012_prediction_WTI_layered, Rasche2013_stacked_topological_bismuth},
while translation symmetry defects called dislocations can trap topological bound states \cite{ran_weak_2010, hu2022dislocationmajoranaboundstates, Mross_2016, Slager_2014, Xue_2021experimentaldislocation}. Weak SPTs
also have the potential for hosting helical edge states \cite{Mong_2012_quantum_transport, Zhong_2023_WTIspinhallchannels, ImuraTakaneTanaka_2011_WTIgaplesshelical, Wang2021_helicalhingestates,sbierski_weak_2016} and even producing non-abelian anyons~\cite{Mross_2016, hu2022dislocationmajoranaboundstates}. Weak SPTs have been studied experimentally in e.g.\ \cite{hamasaki2017dislocation, liu2018experimental, zhang2021observation}.

The question of whether weak SPTs persist in the presence of disorder or short-range interactions is an active area of research \cite{Ringel_2012_strongsideofweak, claes_disorder_2020, Barkeshli_2012, Wang12_InteractingInsulators, hughes_weak_2015} and is important for experimental study, as interactions are not always able to be experimentally controlled.
Using our free-to-interacting map, we mathematically 
{predict}
the stability of weak SPTs under short-range interactions in low dimensions for all tenfold-way symmetry types, and additionally predict the existence of some intrinsically-interacting weak SPTs. Our results are collected in \S\ref{example_section}, and we highlight a few of them here:
\begin{itemize}
    \item We find that the three $\Z$-valued indices of (3+1)d free weak phases in class BDI remain nontrivial under interactions. This agrees with the physical results of Xiao-Kawabata-Luo-Ohtsuki-Shindou~\cite{xiao_anisotropic_2023}.
    \item In class A, we find that the Hall conductivity, an invariant of (3+1)d free weak phases, is stable under interactions. This gives a mathematical derivation of the results of Varjas-de Juan-Lu~\cite[\S II]{varjas_space_2017}.
    \item In class CI, we find a $\Z_2^4$ of interaction-enabled weak phases in dimension (3+1). These phases have not appeared in the literature to our knowledge, and it would be interesting to study them further.
\end{itemize}
We also study the remaining seven Altland-Zirnbauer classes, indicating where our computations mathematically support preexisting results in the literature, and where our predictions are new.

Having discussed our motivation and results, we next expand upon our mathematical methods, which include $K$-theory and TQFT.
The relation between $K$-theory and condensed matter first appeared when generalizing the integer quantum Hall effect (IQHE) to include magnetic disorder, as described in the work of Bellissard \cite{bellissard_disorder_1986}. With the advent of new topological phases such as the topological insulator, the $K$-theory framework for classifying free-fermion phases in the tenfold way was introduced contemporaneously by Kitaev
\cite{kitaev_periodic_2009} and Ryu-Schnyder-Furusaki-Ludwig \cite{Ryu_2010_TIandS},
and later expanded upon by Freed-Moore~\cite{freed_twisted_2013}, Thiang~\cite{thiang2016k}, Alldridge-Max-Zirnbauer \cite{alldridgezirnbauer}, and others. In its simplest version, for a periodic system representing a crystal without any disorder and no interactions, applying the Bloch transform to the associated Hamiltonian $\mathcal{H}$ yields a family of Bloch Hamiltonians $\{\mathcal{H}(\vec{k})\}$ parametrized by the crystalline momentum $\vec{k}$ in the Brillouin torus. 
If $\mathcal{H}$ has a gap in its spectrum, one can interpret its corresponding family as both representing a gapped phase (where a gapped phase is a collection of systems that can be adiabatically connected) and an element of a $K$-theory group of the crystalline momentum space (where adiabatic connection is now interpreted as a homotopy between different families of the Bloch Hamiltonians).
According to whether the Hamiltonian has time reversal symmetry, charge-conjugation symmetry, or chiral symmetry, it is assigned to one of ten Altland-Zirnbauer symmetry classes, corresponding to one of the two shifts of complex $K$-theory (A and AIII) or eight shifts of real $\KO$-theory (D, BDI, AI, CI, C, CII, AII, DIII). That there are only ten such cases for the $K$-theory of a point follows from the Bott periodicity theorem.

In \S\ref{sec:freeKth}, we present a more general $K$-theory framework in a way that makes this correspondence between $K$-theory and symmetry generators precise, using the language of Clifford algebras and group $C^*$-algebras. In \cref{def:neutralSPT}, we give our mathematical model for the group of free fermionic phases with symmetry $C^*$-algebra $A$, using work of one of us~\cite{neutralluuk}. 
\Cref{neutral_ansatz} shows that in the examples corresponding to the tenfold way, these $K$-theory groups recover the $\KR$-groups of a point (strong free phases) or the Brillouin torus (weak free phases).
 
\color{black}


The interacting classification using TQFTs has a different mathematical flavor.
Freed-Hopkins \cite{freed_reflection_2021, freed_invertible_2019}, inspired by previous work of Kapustin and collaborators \cite{kapustin2014symmetry, Kapustin_SPT}, classify (strong) interacting SPT phases in a two-step process:
\begin{enumerate}
    \item To each class of the tenfold way, labeled by a number $s$,
    Freed-Hopkins associate a family of Lie groups $H_d(s)$ ($H^c_d(s)$ in the two complex cases), one in each spacetime dimension $d$, which are the symmetry groups of Euclidean-signature spacetime corresponding to the symmetries of that tenfold way class.
    \item Then, they propose {the ansatz} that interacting SPTs are classified by the deformation classes of their low-energy limits, which are postulated to be in one to one correspondence with deformation classes of reflection-positive invertible field theories on manifolds with $H_d(s)$-structure.
\end{enumerate}

Topological quantum field theories (TQFTs) as first axiomatized by Atiyah \cite{atiyah_topological_1988} are rules that assign to a $d$-dimensional manifold $N$ (without boundary) a finite dimensional vector space $V_{N}$ (the state space of the theory) and to a $(d+1)$-dimensional manifold $M$ that has $N_1$ and $N_2$ as its boundaries, 
a linear map $U_{M}\colon V_{N_1}\xrightarrow[]{} V_{N_2}$. Invertible topological field theories are 
those which are invertible under a tensor product operation, forcing the vector spaces in question to be one-dimensional and these theories to be classified by groups. In condensed matter, the most well-known applications of TQFT are $2+1$d $\U(1)$ Chern-Simons field theory applied to the IQHE and the $3+1$d axion topological $\theta$ term for topological insulators \cite{fradkin_field_2013}. 

The connection to TQFTs starts from the gap condition on the spectrum of the Hamiltonian $\mathcal{H}$ of a many-body interacting system, which allows one to treat the finite-dimensional ground state space separately from excited states.
For SPT phases, one only considers systems that have a unique ground state; i.e.\ to each system of interest we can assign a one-dimensional vector space corresponding to the ground state of its Hamiltonian $\mathcal{H}$, whose low-energy limit is postulated to be part of an invertible topological field theory with the above $H_{d}(s)$-structure \cite{freed_invertible_2019, freed_reflection_2021}.
This postulate, which concerns the renormalization group flow of certain lattice models to their low-energy limits, is an area of active research~\cite{Gaiotto, RowellWang, saberi_spin_2018, Fre19, freed_reflection_2021, freed_topological_2022,Chen_SPT}. The existence of free-to-interacting maps provides additional evidence for this postulate, and in fact for the specific classification of Freed-Hopkins, as we discuss in \S\ref{supercohomology_is_wrong}.

Under the above assumptions, interacting SPTs of a prescribed dimension and tenfold way class are then classified by a group of invertible field theories.
Freed-Hopkins-Teleman~\cite{FHT10} and Freed-Hopkins~\cite{freed_reflection_2021} compute these groups of invertible field theories in terms of bordism groups.

Our work concerns the connection between the free and interacting classifications of SPTs, and thus the connection between $K$-theory and invertible field theories. Roughly speaking, the free-to-interacting map of Freed-Hopkins is a homomorphism constructed using the index of a twisted Dirac operator. The computation of this index generalizes the original cases discussed in Atiyah-Bott-Shapiro \cite{atiyah_clifford_1963}, which correspond to class A and class D systems, to all tenfold way types.
See \cite[\S9.3]{freed_reflection_2021}.

Our contribution is the generalization of this map to the case of discrete translation-invariant phases, including weak phases.
Free weak phases are computed from the $K$-theory of the Brillouin zone torus, as derived in many previous works including Kitaev's periodic table paper \cite{kitaev_periodic_2009}.
Interacting weak phases, on the other hand, have not been directly studied homotopically. Inspired by work of Freed-Hopkins, specifically Example 2.3 of \cite{freed_invertible_2019}, we decide to model interacting phases as a family of SPTs parameterized by the spatial (i.e.\ unit cell) torus.
Our generalization of the free-to-interacting map to weak SPTs involves two innovations: the use of T-duality, which exchanges the spatial and momentum tori in $K$-theory (\S\ref{subsec:T-duality}), and the binomial formula for the cohomology of the torus (\S\ref{tori_splitting-section}), which reflects the fact that weak invariants can be constructed from strong ones.
See~\cite{Mathai_2016-Tdualitybulkboundaryhigherdim,freed_k-theory_2016} for prior work on T-duality in this context, and \cite{Mathai_2016-Tdualitybulkboundaryhigherdim,freed_k-theory_2016,Xiong_SPT, stehouwer_k-theory_nodate} for the binomial formula.

Throughout this paper, we introduce concepts by first implementing them concretely in class AII, corresponding to time-reversal invariant topological insulators, which are the most well-studied class of weak SPTs. This in particular includes Appendix~\ref{app:twABS}, where we compute Freed-Hopkins' class AII index map in dimension $4$. Freed-Hopkins showed this map is surjective~\cite[Corollary 9.93]{freed_reflection_2021}, and we give another proof which also allows us to explicitly calculate the map on example manifolds.

We intend this work to be an introduction to free-to-interacting maps as well as a full computation of the maps in the case of weak
phases for the tenfold way. Ongoing work of some of the authors addresses free-to-interacting maps for more general crystalline symmetries. 
\stoptoc
\subsection*{Outline}
In \S\ref{ansatz...} we provide our version of the $K$-theory classification of free SPTs and review the {low-energy IFT ansatz and the} bordism classification of interacting SPTs, starting from the idea of a fermionic symmetry group.
We get to the strong free-to-interacting map in \S\ref{freed_hopkins_F2I}, then define the new weak counterpart of this map in \S\ref{s:second_part_section_2}. The precise definition of the weak free-to-interacting map is given in \cref{the_weak_F2I_ansatz}.
In \S\ref{example_section}, we compute our weak free-to-interacting map in spatial dimensions 1, 2, and 3 for all ten Altland-Zirnbauer types and compare with existing literature.

\subsection*{Acknowledgments}
We thank Dan Freed, Mike Hopkins, Ralph Kaufmann, and Ryan Thorngren for helpful discussions on the content of this paper, and an anonymous referee for helpful comments, including catching a gap in one of our arguments. CK and NPT were supported by the National Science Foundation under Grant No. DGE-2141064 during the writing of this paper. 
LS is grateful to the Atlantic Association for Research in the Mathematical Sciences and the Simons Collaboration on Global Categorical Symmetries for its financial support and Dalhousie University for providing the facilities to carry out his research.
CK also thankfully acknowledges support of the Simons Collaboration on Global Categorical Symmetries as well as the Killam Trusts.

\resumetoc


\section{Classifying invertible phases}
    \label{ansatz...}

In this section we introduce a mathematically precise framework {that under the low-energy IFT ansatz provides a classification for}
symmetry-protected topological phases. Our goal is twofold:
first, to set up rigorous models for the symmetry data relevant to fermionic systems, and
second, to explain how the resulting classifications can be computed in practice. We start
by introducing fermionic symmetry groups and their associated superalgebras in \cref{sec:fermsymmetry}, then describe
the free classification via $K$-theory within that context in \cref{sec:freeKth}.
Afterwards in \cref{subsec:bordismclassification} we introduce the interacting classification via bordism focusing on the approach of Freed and Hopkins. This will
ultimately allow us to formulate the free-to-interacting comparison map for strong phases in \cref{freed_hopkins_F2I}, a prerequisite for our analysis of the analogous map for weak phases.

\subsection{Fermionic symmetry groups}
\label{sec:fermsymmetry}

\begin{description}
\item[Fermionic groups] 
In condensed matter, symmetry groups $G_f$ of fermionic systems have the extra structure of what we will call a \emph{fermionic group}~\cite{Ben88, stehouwer_interacting_2022}.
This means that $G_f$ comes equipped with a central element $(-1)^F \in G_f$ of order two and a homomorphism $\phi\colon G_f \to \Z_2$ labeling time-reversing elements such that $(-1)^F$ is time-preserving.
    \item[Symmetries form a superalgebra] A \term{superalgebra} is a $\Z_2$-graded algebra. Each Altland-Zirnbauer class specifies a set of symmetry operators, which generate a superalgebra over $\R$ or $\C$. The reader should be warned that the physical interpretation of the $\Z_2$-grading here is given by time-reversing versus time-preserving symmetries, as opposed to fermions versus bosons.
    \item[From superalgebra to fermionic group] The superalgebras we obtain by the Altland-Zirnbauer classification are \term{super division algebras}, meaning all homogeneous elements are invertible. There are exactly ten such superalgebras~\cite{wall1964graded}
    \begin{equation}
D^\C_i,D^\R_j \quad i=0,1, ~ j=0,\dots, 7,
\end{equation}
which can be constructed explicitly as certain Clifford algebras.
Given a super division algebra $A$, the set $S(A)$ of norm-$1$ elements of $A$ acquires the structure of a compact Lie group from the multiplication on $A$. The grading operator defines a homomorphism $\phi\colon S(A)\to\Z_2$, and $(-1)^F$ generates a central $\Z_2$ subgroup of $S(A)$, making $S(A)$ into a fermionic group.
\end{description}

\begin{example}\label{AII_first_exm}
    The Altland-Zirnbauer class AII, corresponding to topological insulators with a time-reversal symmetry, has a time-reversal symmetry squaring to $(-1)^F$ and a charge $Q$ generating a $\U(1)$ symmetry corresponding to conservation of particle number. These symmetries are subject to a \term{spin-charge relation}: the $-1$ in this $\U(1)$ is equal to $(-1)^F$, and time-reversal acts on $\U(1)$ by complex conjugation.

The algebra generated by $T$ and $Q$ over $\R$ is isomorphic to the Clifford algebra
\begin{equation}
    \Cl_{-2}\coloneqq\R\langle e_1, e_2 \rangle/(e_1^2 = e_2^2 = -1, e_1e_2 + e_2e_1= 0).
\end{equation}
This isomorphism can be explicitly realized by sending $T\mapsto e_1$ and $e^{i\pi\theta}\in\U(1)$ to $\cos(\theta) + e_1e_2\sin(\theta)$~\cite[Example 4]{stehouwer_interacting_2022}.

The second step is to find $S(\Cl_{-2})$, which by definition of the pin groups is equal to $\Pin^-(2)$. First consider the real superalgebra
    \begin{equation}
        A'\coloneqq \C[T]/(T^2 = -1, iT + Ti = 0),
    \end{equation}
    where $i$ is even and $T$ is odd. There is an isomorphism $\phi\colon A'\to \Cl_{-2}$ of real superalgebras defined by setting $\phi(i) = e_1e_2$ and $\phi(T) = e_1$, then extending linearly to all of $A'$.

    From the viewpoint of $A'$, it is easier to find the homogeneous norm-one elements: the unit complex numbers, which generate a $\U(1)$ subgroup of $S(A')$, and the $\Z_4$ subgroup generated by $T$. The operator 
    $T$ acts on $\U(1)$ by complex conjugation, and $T^2 = -1$ is in $\U(1)$, so we see that 
    \begin{equation}
    S(\Cl_{-2}) \cong S(A')\cong \frac{\U(1) \rtimes \Z_4^T}{\Z_2}.
    \end{equation}
    The homomorphism $\phi$ is the unique one which is trivial when pulled back to $\U(1)$ and nontrivial when pulled back to $\Z_4^T$; $(-1)^F$ is the common central element.
\end{example}

\subsection{\textit{K}-theory classifications of free fermion phases}
\label{sec:freeKth}

The classification of SPT phases of complex free fermions can be connected to $K$-theory as follows \cite{freed_twisted_2013}.
For a symmetry group $G$, consider a one particle state space $V$, which furnishes a representation $R$ of $G$.
We want to understand the space of all gapped Hamiltonians $H$ on $V$ with symmetry $G$.
After shifting the Fermi energy to zero, a gapped Hamiltonian is defined as a linear operator without kernel that intertwines $R$.
This splits the representation $V = V_{\mathit{valence}} \oplus V_{\mathit{conduction}}$ into $\pm$-eigenspaces of $H$.
Therefore $H$ defines an element 
\begin{equation}
V_{\mathit{valence}} - V_{\mathit{conduction}} \in K_G^0(\pt)
\end{equation}
in the representation ring of $G$. 
    If two Hamiltonians give different elements of $K_G^0(\pt)$, a path between them must involve crossing the gap.
Conversely, two different Hamiltonians with the same decomposition $V = V_{valence} \oplus V_{conduction}$ are in the same path component by spectral flattening.
Therefore, the set of components $\pi_0$ essentially\footnote{In the representation ring, we quotient out by additional relations such as $V - V = 0$ to ensure $K_G^0(\pt)$ is a group. A priori, there is no physical justification for requiring this invertibility under stacking (which is given by direct sum since we are on a $1$-particle space). However, here we restrict to invertible phases. Phases which are unstably nontrivial are called \emph{fragile} phases~\cite{po2018fragile,else2019fragile}.} equals $K_G^0(\pt)$.

With enough care about the mathematical details, the above heuristic applies in various settings:
\begin{enumerate}
    \item When $G$ contains time-reversal symmetries, they act anti-unitarily on $V$. We have to accommodate for this in the definition of the representation ring.
    \item Some symmetries have additional constraints relating to fermion parity, such as $T^2 = (-1)^F$ for a time-reversal symmetry. 
    Since $(-1)^F$ acts by $-1$ on $V$, we have to enforce this relation in the representation ring.
    \item In positive spatial dimension $d$, it is reasonable to require $G = \Z^d$ to be the symmetry group of a lattice of atoms.
    Equivariant stable homotopy theory for noncompact $G$ is still in development.\footnote{Though see~\cite{barcenas_equivariant_2014,degrijse_proper_2023,linskens_global_2025}.} Since the group algebra of an infinite group will not suffice for these purposes, we define $K_G^0(\pt)$ to be the $K$-theory of the complex group $C^*$-algebra $C^*(G)$ of $G$.
\end{enumerate}

As argued by \cite[Example 9.1-9.3]{thiang2016k}, there is an isomorphism $K^0_{\Z^d}(\pt) \cong K^0(\mathbb{T}^d)$, closely related to Bloch's theorem. 
Here $K^0(\mathbb{T}^d)$ is the $K$-theory of the Brillouin zone torus.
The isomorphism is given by a Fourier transform to momentum space, a special case of the Pontryagin duality isomorphisms of $C^*$-algebras
\begin{equation}
C^*(\Z^d) \cong C(\widehat{\Z^d}, \C) = C(\mathbb{T}^d, \C).
\end{equation}
Here $C(X,\C)$ denotes the ring of continuous functions on $X$ and $\widehat{\Z^d} \coloneqq \Hom(\Z^d, U(1)) = \mathbb{T}^d$ is the Pontryagin dual of $\Z^d$.
Explicitly, a vector bundle $E$ over $\mathbb{T}^d$ gives a $C^*(\Z^d)$-module $\Gamma(E)$ of continuous sections of $E$ by mapping $\vec{n} \in \Z^d$ to the function $\mathbb{T}^d \to \C$ given by $e^{i \vec{n} \cdot \vec{k}}$.
Here we used the common convention of identifying $\mathbb{T}^d$ with a quotient of the box $[-\pi, \pi]^d$ using the map $k \mapsto (\vec{n} \mapsto e^{i \vec{n} \cdot \vec{k}})$.
We have therefore reproduced the fact \cite{kitaev_periodic_2009} that class A topological insulators in spatial dimension $d$ are classified by $K^0(\mathbb{T}^d)$.

In order to address the question of which topological phases survive in the continuum limit, we redo the above argument for $G = \R^d$ the group of continuous translations.
We again have the Fourier transformation isomorphism
\begin{equation}
C^*(\R^d) \cong C(\widehat{\R^d}, \C) = C(\R^d, \C),
\end{equation}
so that 
\begin{equation}
K^0_{\R^d}(\pt) = K_0(C^0(\R^d)) = \widetilde{K}{}^0(S^d).
\end{equation}
This agrees with the classification of strong class A topological insulators.

In the above discussion, we implicitly assumed that our fermions are charged. 
In other words, we assumed the existence of a polarization giving the one particle spaces $V$ and $V^*$ of creation and annihilation operators, thus disallowing unpaired Majoranas. 
There is an analogous discussion for neutral fermions, resulting in $\KO$- instead of $K$-theory.
This approach can be formulated in the Bogoliubov-de-Gennes formalism \cite{alldridgezirnbauer}. 
Even though most of the condensed matter literature does not use Majorana fermions, we will focus on this perspective, following our main references \cite{kitaev_periodic_2009} and \cite{freed_reflection_2021}.

The main difference in the new setup will be that the complex one particle Hilbert space $V$ is replaced by a real Hilbert space $\mathcal{M}$. 
The self-adjoint gapped Hamiltonian $H$ is replaced by a skew-adjoint gapped operator $\Xi$ on $\mathcal{M}$, which one should think of as $-iH$.

Because $\Xi$ is skew-adjoint, its spectrum consists of elements of the form $ai$ where $a$ is a nonzero real number (nonzero because $\Xi$ is invertible). We can thus define $\Xi/\lvert\Xi\rvert$ by functional calculus. Since the square of the function $z\mapsto z/\lvert z\rvert$ defined on $\C^\times$ takes the value $-1$ on numbers of the form $ai$ with $a\in\R^\times$, we deduce that $(\Xi/\lvert\Xi\rvert)^2 = -1$: $\Xi/\lvert\Xi\rvert$ defines a complex structure on $\mathcal M$. Stably, the space of almost complex structures on $\mathcal M$ becomes a classifying space for $\KO^{-2}$,
%
hinting towards a relationship between neutral phases and $\KO$-theory.
This discussion generalizes to arbitrary symmetry groups, taking into account that time-reversal symmetries should anti-commute with $\Xi$.

We can use the formalism of Karoubi triples~\cite{thiang2016k,donovankaroubigraded} to make this discussion mathematically precise: let $A$ be the (real or complex) super $C^*$-algebra of symmetries, graded by time-reversal.\footnote{In this work, we will restrict to the case where $A$ is the tensor product of a tenfold way symmetry as explained in \S \ref{sec:fermsymmetry} with the group $C^*$-algebra of the Lie group of translation symmetries, either discrete $\Z^d$ or continuous $\R^d$.} 
    A \emph{Karoubi triple} $(\mathcal{M},\Xi_1, \Xi_2)$ consists of a finitely generated (ungraded) $A$-module $\mathcal{M}$ and maps $\Xi_i\colon \mathcal{M} \to \mathcal{M}$ satisfying $\Xi^2 = - \id_\mathcal{M}$ and $\Xi_i a = (-1)^{|a|} a \Xi_i$ for all $a \in A$.\footnote{There is an infinite dimensional version of the Karoubi description in which instead one assumes a finiteness condition on $\Xi_2-\Xi_1$. This description can be shown to be equivalent to the finitely generated version currently under discussion \cite[\S 4]{gomi_freed-moore_2021}. Using modules that are not finitely generated can be more suitable for physics, for example if we want to take the unbounded above valence band into account.}
    One can think of a Karoubi triple as a formal difference $[\Xi_1] - [\Xi_2]$ of Hamiltonians with $A$-symmetry.
    We now want to impose that $[\Xi_1] - [\Xi_2] = 0$ if $\Xi_1$ and $\Xi_2$ are in the same path component.
    So define a Karoubi triple to be \emph{elementary} when $\Xi_1$ is in the same path component as $\Xi_2$ in the space of complex structures $\Xi$ such that $\Xi a = (-1)^{|a|} a \Xi$ for all $a \in A$.
    Two Karoubi triples $(\mathcal{M},\Xi_1, \Xi_2), (\mathcal{M}',\Xi_1', \Xi_2')$ are \emph{isomorphic} if there exists an $A$-module isomorphism $\mathcal{M} \to \mathcal{M'}$ intertwining $\Xi_i$ with $\Xi_i'$ for $i=1,2$.
    Note that there is an obvious notion of direct sum $\oplus$ of Karoubi triples.
    We say two triples $T_1,T_2$ are \emph{stably equivalent} if there exist an elementary triples $T_1',T_2'$ such that $T_1 \oplus T_1'$ is isomorphic to $T_2 \oplus T_2'$.
    The set of Karoubi triples can be thought of as a stabilization of the space of $A$-symmetric Bogoliubov de-Gennes Hamiltonians $\Xi$.

    \begin{definition}
    \label{def:neutralSPT}
    The group of \emph{$A$-symmetric free SPT phases} is the set of Karoubi triples modulo stable equivalence under $\oplus$.

    If $A = C^*(\Z^d;F)\otimes\Cl_{-s}$, where $F = \R$, resp.\ $F = \C$, we will refer to $A$-symmetric free SPT phases as \term{discrete translation-invariant free SPT phases of real, resp.\ complex Altland-Zirnbauer class $s$.}
    Similarly, $A = C^*(\R^d;F)\otimes\Cl_{-s}$ gives \term{continuous translation-invariant free SPT phases.}\footnote{Because $\Cl_{-s}$ is finite-dimensional, different notions of $C^*$-tensor product agree. 
    }
    \end{definition}
\begin{theorem}[\cite{neutralluuk}]
\label{neutraltheorem}
     The group of $A$-symmetric free SPT phases is isomorphic to $\KO_{2}(A)$.
\end{theorem}

\begin{remark}
    Suppose $A$ is a real super $C^*$-algebra containing a subalgebra $\C$, which is not necessarily in the center.
    We think of this subalgebra as generating charge.
    Suppose additionally that $A = A_+ \oplus A_-$ where $a_{\pm} \in A_{\pm}$ if and only if $a_{\pm} z = z^{\pm} a_{\pm}$ for all $a \in A_{\pm}$ and $z \in \C$.
    This defines a $\Z_2$-grading $\mu$ on $A$ not necessarily equal to the $\Z_2$-grading $\phi$ given by time-reversal.
    Note that these this grading commutes with the other $\Z_2$-grading on $A$ in the sense that the corresponding operators with eigenvalues $\pm 1$ commute.
    Therefore, there is a product/diagonal $\Z_2$-grading $c$.
    Then $\KO_0(A,c) \cong \KO_2(A,\phi)$,
    where $K_i(A,\lambda)$ denotes the degree $i$ $K$-theory of the algebra $A$ with $\Z_2$-grading $\lambda$.
    This connects the description of Theorem \ref{neutraltheorem} to the discussion of the beginning of this section and in particular to \cite{freed_twisted_2013}.
\end{remark}

\begin{example}
\label{AIIKtheory}
\label{AII_d3_Ktheory_section}
    Take $A = C^*(\Z^d;\R) \otimes \Cl_{-2}$ to be the tensor product of a $d$-dimensional discrete translation symmetry and the internal symmetry algebra of class AII (see Example \ref{AII_first_exm}).
    Using the fact that $K_i(A \otimes \Cl_{\pm 1}) \cong K_{i \mp 1}(A)$~\cite{karoubi1968algebres}, we see that the group of $A$-symmetric free SPT phases is given by $\KO_{2}(A) \cong \KO_{4}(C^*(\Z^d;\R))$.
    We can now apply arguments as above to relate $\Z^d$ to the torus, but there is one important subtlety.
Namely, the Fourier transform crucially uses the complex numbers through the factor $e^{i \vec{n} \cdot \vec{k}}$ and 
\begin{equation}
\ol{e^{i \vec{n} \cdot \vec{k}}} = e^{i \vec{n} \cdot (-\vec{k})}.
\end{equation}
So under the isomorphism $C^*(\Z^d;\C) \cong C(\mathbb{T}^d; \C)$ of complex $C^*$-algebras, complex conjugation on the left hand side gets mapped to the operation mixing complex conjugation with the involution $k \mapsto -k$ on the Brillouin zone.
We will let $\overline{\mathbb{T}}{}^d$ denote the Brillouin torus with this involution; likewise, we will let $\bar S^d$ denote $S^d = \R^d\cup\{\infty\}$ with the involution $k\mapsto -k$, fixing the point at infinity.
Then there is an isomorphism of complex $C^*$-algebras with Real structure $C^*(\Z^d;\C) \cong C(\overline{\mathbb{T}}{}^d; \C)$.
Therefore the $K$-theory of $C^*(\Z^d;\R)$ is not the $\KO$-theory of the torus, but its $\KR$-theory, which depends on this involution.
We obtain that the classification of class AII topological insulators is given by
\begin{equation}
\label{AII_Cstar_alg}
\KO_{4}(C^*(\Z^d;\R)) \cong \KR_{4}(C(\overline{\mathbb{T}}{}^d;\C)) \cong \KR^{-4}(\overline{\mathbb{T}}{}^d).
\end{equation}
Replacing $\Z^d$ by a continuous translation symmetry $\R^d$, we obtain similarly that
\begin{equation}
\KO_{4}(C^*(\R^d;\R)) \cong \widetilde{\KR}{}^{-4}(\bar{S}{}^d).
\end{equation}

For example, consider the $d=3$ time-reversal invariant insulator in class AII studied first in \cite{fu_topological_2007-1,fu_topological_2007-2}. 
We classify its phases using \eqref{AII_Cstar_alg} as
    \begin{equation}
        K_2(C^*(\Z^3) \otimes \Cl_{+2}) \cong \KR^{\textcolor{black}{-}4}(\overline{\mathbb{T}}{}^3) \cong \Z \oplus \Z_2 \oplus (\Z_2)^3,
    \end{equation}
    see \cite[Theorem 11.14]{freed_twisted_2013}.
    As observed in e.g. \cite{kitaev_periodic_2009} and \cite[Theorem 3.35]{freed_k-theory_2016}, one $\Z_2$ invariant encodes the strong phase detected by the Fu-Kane-Mele invariant. 
    The $\Z$ invariant counts the number of Kramers pairs of electrons, and the $(\Z_2)^3$ vector invariant encodes the weak topological phases: phases protected by the discrete translation symmetry. These phases may be viewed as quantum spin Hall phases living on each two-dimensional cross section
    of the three-dimensional material.
    Indeed, for a continuous translation symmetry, we obtain $\widetilde{\KR}{}^{-4}(\bar{S}{}^3) \cong \Z_2$ and only the first $\Z_2$ survives.
\end{example}

Example \ref{AIIKtheory} generalizes to all tenfold classes to obtain the following corollary of Theorem \ref{neutraltheorem}:

\begin{corollary} 
\label{neutral_ansatz}
Discrete (resp.\ continuous) translation-invariant free SPT phases of real Altland-Zirnbauer class $s$ in spatial dimension $d$ are classified by $\KR^{s-2}(\overline{\mathbb{T}}{}^d)$ (resp.\ $\widetilde{\KR}{}^{s-2}(\bar{S}{}^d)$).
There is a similar statement for the two complex classes, replacing $\KR$ by complex $K$-theory and forgetting the involutions on $\overline{\mathbb{T}}{}^d$ and $\bar S{}^d$.
\end{corollary}

\begin{remark}
    \label{rem:numberofbands}
    In our convention, class A weak SPT phases are classified by \emph{unreduced} $K$-theory $K^0(\mathbb{T}^d)$.
    In the decomposition $K^0(\mathbb{T}^d) \cong K^0(\pt) \oplus \widetilde{K}{}^0(\mathbb{T}^d)$, the first term corresponds to the $0$-cell of the Brillouin zone.
    Physically, this $K^0(\pt) \cong \Z$-valued invariant is a comparison count of the number of bands below versus above the gap.
    An analogous argument applies to the other classes, where the invariant can also be $\Z_2$-valued or nonexistent depending on $\KO^{s-2}(\pt)$.
    This invariant is typically ignored in the condensed matter literature, but we would argue it should be included as a weak phase corresponding to a $0$-dimensional strong phase.
\end{remark}

\subsection{Bordism classifications of interacting phases}
\label{subsec:bordismclassification}

As mentioned in the introduction,
when we ``turn on interactions'' by regarding free fermion Hamiltonians in the context of all symmetry-protected gapped lattice Hamiltonians {with no ground-state degeneracy}
(i.e.\ representatives of invertible topological phases), it is conjectured that deformation classes of invertible topological phases are classified by their low-energy behavior, captured by a reflection-positive invertible field theory (IFT). See \cite{Fre19, freed_reflection_2021} 
for further discussion of this ansatz, which is supported by a strong body of computational evidence \cite{freed_reflection_2021, Cam17, KT17, BC18, WG18, Gaiotto_SPT, freed_invertible_2019, ABK21, debray_invertible_2021, BCHM22}.

\subsubsection{The classification of reflection-positive invertible field theories}
\label{FH_IFT}
Reflection-positive IFTs are defined at a mathematical level of rigor by Freed-Hopkins~\cite[\S 8]{freed_reflection_2021} in the topological case and Grady-Pavlov~\cite[\S 5]{GP21} in the nontopological case.\footnote{See~\cite{JF17, DAGGER24, CFHPS24, Ste24, MS23} for more about reflection positivity in the noninvertible setting.} They are classified using generalized cohomology; before we give the classification in \cref{IFT_class}, we review some key definitions.

Let $\O$ denote the infinite orthogonal group $\operatorname{colim}_n\O(n)$, and let $\rho\colon H\to \O$ be a homomorphism of topological groups. Given this data, we will let $H_d$ be the pullback of $\rho$ along $\O(d)\hookrightarrow\O$.
\begin{definition}
Given $\rho\colon H\to\O$ as above, let $\Omega_*^H(\text{--})$ denote the generalized homology theory called \term{$H$-bordism}: $\Omega_n^H(X)$ is the abelian group of closed $n$-manifolds with an $H$-structure~\cite{chern1966geometry} and a map to $X$ under disjoint union, modulo bordisms of such data.
\end{definition}
For example, $\Omega_*^\SO$ is the bordism theory of oriented manifolds.
\begin{definition}\label{anderson_dual}
There is a duality on generalized homology and cohomology called \term{Anderson duality}~\cite{And69, Yos75}. Given a generalized homology theory $E_*$, the \term{Anderson dual of $E_*$} is the generalized cohomology theory $(I_\Z E)^*$ defined to satisfy the following universal property: for all spaces $X$, there is a natural short exact sequence
\begin{equation}
\label{IZproperty}
\shortexact{\Ext(E_{n-1}(X), \Z)}{(I_\Z E)^n(X)}{\Hom(E_n(X), \Z)}.
\end{equation}
\end{definition}
One can check that~\eqref{IZproperty} actually uniquely characterizes a generalized cohomology theory $(I_\Z E)^*$. Moreover, because $\Hom$ and $\Ext$ are contravariant functors in their first argument, Anderson duality defines a contravariant functor on cohomology theories: given a natural transformation $E_*(\text{--})\Rightarrow F_*(\text{--})$, there is a natural transformation $(I_\Z E)^*(\text{--}) \Leftarrow (I_\Z F)^*(\text{--})$.

The short exact sequence~\eqref{IZproperty} splits, but \emph{not} naturally,\footnote{This is a generalization of the unnatural splitting of the short exact sequence in the universal coefficient theorem~\cite[\S 4]{And69}.} implying an isomorphism from $(I_\Z E)^n(X)$ to the
direct sum of the torsion subgroup of $E_{n-1}(X)$ with the free part of $E_n(X)$.

For more
on $I_\Z$ and its appearance in this context, see Freed-Hopkins~\cite[\S 5.3, 5.4]{freed_reflection_2021}.

\begin{theorem}[{Freed-Hopkins~\cite[Theorem 1.1]{freed_reflection_2021}, Grady~\cite[Theorem 1]{Gra23}}]
\label{IFT_class}
Let $\mho_H^*$ denote the Anderson dual cohomology theory to $\Omega_*^H$. Then there is a natural isomorphism from the abelian group of deformation classes of $d$-dimensional IFTs on manifolds with $H$-structure to $\mho_H^{d+2}$.
\end{theorem}
As always, $d$ is the spatial dimension of the theory.
\subsubsection{Spacetime symmetry groups for the tenfold way}\label{spacetime_sym}
\Cref{IFT_class} leads us to use the cohomology theory $\mho_H^*$ to model interacting phases, but we need to determine $H$ and its map to $\O$ for the ten collections of symmetries we are interested in. The reference~\cite[\S 3.2 and 3.3]{stehouwer_interacting_2022} provides a unified way of doing this.

  There is a construction of a spacetime structure group $H(G)$ from an internal symmetry group $G$ indicated in \cite{freed_reflection_2021}; see \cite{stehouwer_interacting_2022} for a construction based on \cite{stolz_preprint}. Given a fermionic group $(G, \phi, (-1)^F)$, 
  one first takes the central product
    \begin{equation}
        \widetilde H \coloneqq \frac{G\times \Pin^-}{\langle((-1)^F,-1)\rangle},
    \end{equation}
    where $-1 \in \Pin^-$ is the nontrivial element in the kernel of the map to $\O$.
    There is a homomorphism $\widetilde\phi\colon \widetilde H\to\Z_2$ defined by sending $(g, B)\in G\times\Pin^-$ to $\phi(g) + \det(B)$, where $\det\colon\Pin^- \to\Z_2$ corresponds to the homomorphism taking the $\{\pm 1\}$-valued determinant of $B$ under the canonical isomorphism $\set{\pm 1}\cong\Z_2$.

    Finally, the tangential structure $H(G)$ associated to $G$ is the group $\widetilde\phi^{-1}(0)$, with the map to $\O$ induced by the map on $\Pin^-$. 
    It is easy to show that the corresponding family of topological groups $H_d(G)$ is obtained by replacing $\Pin^-$ by $\Pin^-(d)$ in the above discussion.
    
\begin{proposition}\label{pin_down_tangential_structure}
Let $G$ be a fermionic group with $\phi = 0$ trivial and let $i\colon G \to G$ be an involution.
Define the two fermionic groups
\begin{equation}\label{G_plus_minus}
G_{\pm} \coloneqq \frac{G \rtimes \Pin^{\pm}(1)}{\Z_2^F},
\end{equation}
where the semidirect product is defined using $i$ and $\det\colon \Pin^\pm(1) \to \Z_2$.
The $\phi$ is defined by projection onto the second factor.
Then there is an isomorphism of fermionic structure groups
\begin{equation}
H_d(G_{\pm}) \cong \frac{G \rtimes \Pin^{\mp}(d)}{\Z_2^F},
\end{equation}
where the semidirect product is again defined using $i$ and $\det\colon \Pin^{\mp}(d) \to \Z_2$.
\end{proposition}
\begin{proof}
First some notation: denote the canonical odd element $T \in \Pin^{\pm}(1) \subset G_{\pm}$, so $T^2 = (\pm 1)^F$ and $g T = T i(g)$ for $g \in G$ the elements with $\phi(g) = 0$.
Given elements $x_1, x_2 \in \Pin^+(d)$, we define a new group structure (the `graded opposite') by 
\begin{equation}
x_1 * x_2 \coloneqq 
\begin{cases}
(-1)^F x_1 x_2 & \text{ both odd,}
\\
x_1 x_2 & \text{ otherwise.}
\end{cases}
\end{equation}
Then $(\Pin^+(d),*) \cong \Pin^-(d)$ as fermionic groups.

Define the map
\begin{subequations}
\begin{equation}
\psi\colon \frac{G \rtimes \Pin^{\mp}(d)}{\Z_2^F} \to H(G_{\pm}) \subseteq \frac{G_{\pm} \times \Pin^-(d)}{\Z_2^F}
\end{equation}
by 
\begin{equation}
\psi(g \rtimes x) = 
\begin{cases}
(g , x) & \det(x) = 0,
\\
(g T , x) & \det(x) = 1.
\end{cases}
\end{equation}
\end{subequations}
This is well-defined because when we quotient by $\Z_2^F$, $\psi$ lands in $H(G_{\pm})$ because $\det(x) + \phi(T) = 0$ if $\det(x) = 0$, and $\det(x) + \phi(gT) = 0$ if $\det(x) = 1$.
To check that this is a homomorphism, we let $g_1 \rtimes x_1, ~ g_2 \rtimes x_2 \in G \rtimes \Pin^{\mp}(d)$ and have to show $\psi((g_1 \rtimes x_1)( g_2 \rtimes x_2)) = \psi(g_1 \rtimes x_1) \psi(g_2 \rtimes x_2)$. 
There are four cases depending on $\det x_1$ and $\det x_2$.
The most nontrivial case is the one for which both are $1$:
\begin{equation}
(g_1 \rtimes x_1) (g_1 \rtimes x_2) = g_1 i(g_2) \rtimes (\mp 1)^F x_1 x_2,
\end{equation}
where we have used the product $*$ in case we are working in $\Pin^-(d)$ and the normal product of $\Pin^+(d)$ otherwise.
This element is indeed mapped to
\begin{equation}
    (g_1 T , x_1) (g_2 T , x_2) = (g_1 T g_2 T, x_1 x_2) = ( (\pm 1)^F 1 g_1 i(g_2) , x_1 x_2).
\end{equation}
The other three cases are easier.
It is not hard to see that $\psi$ is a bijection.
\end{proof}

\begin{example}
We illustrate how to use \cref{pin_down_tangential_structure} to determine the tangential structures for symmetry classes BDI and DIII, which are the cases $s = 1$ and $s = -1$ respectively. There are isomorphisms of fermionic groups $S(\Cl_{\pm 1})\cong \Pin^\pm(1)$; $\Pin^+(1)\cong\Z_2^F\times\Z_2^T$ and $\Pin^-(1)\cong\Z_4^T$, with $\Z_2^F\subset\Z_4^T$ the unique order-two subgroup. Now apply \cref{pin_down_tangential_structure} with $G = \Z_2$ and the involution $i = \mathrm{id}$: the semidirect product $G\rtimes\Pin^\pm(1)$ simplifies to a direct product, and then $\Z_2$ cancels the $\Z_2$ in the denominator, so in~\eqref{G_plus_minus}, $G_\pm = \Pin^\pm(1)$. In exactly the same way, $H_d(G_\pm)$ simplifies to $\Pin^\mp(d)$. Thus \cref{pin_down_tangential_structure} reproduces a well-known fact in the physics literature: fermionic systems with a time-reversal symmetry $T$ with $T^2 = 1$ correspond to putting pin\textsuperscript{$-$} structures on spacetime, and with $T^2 = (-1)^F$ correspond to putting pin\textsuperscript{$+$} structures on spacetime.
\end{example}

\begin{example}
\label{ex:classAIIstructure}
We come back to class AII. 
In \cref{AII_first_exm}, we obtained the fermionic group $S(\Cl_{-2}) \cong (\Z_4^T\rtimes \U(1))/(\Z_2^F)$ from the symmetry algebra of this class. Using \cref{pin_down_tangential_structure}, we will compute the tangential structure group $H_d(S(\Cl_{-2}))$: there is an isomorphism $\Z_4^T\cong\Pin^-(1)$ of fermionic groups: both have underlying group isomorphic to $\Z_4$ with the map to $\O_1$ nontrivial, and this characterizes $\Z_4^T$ up to isomorphism of fermionic groups. Therefore, we can apply \cref{pin_down_tangential_structure} with $G = \U(1)$ and $i$ equal to complex conjugation. 
Using that $G_- = S(\Cl_{-2})$, we conclude
\begin{equation}
    H_d(S(\Cl_{-2}))\cong \frac{\Pin^+(d) \ltimes \U(1)}{\Z_2^F}.
\end{equation}
Metlitski~\cite[\S III.B]{Met15} introduces this group in the context of invertible phases, and calls it $\Pin_{\tilde c}$. Its appearance in the tenfold way is due to Freed-Hopkins~\cite[(9.9)]{freed_reflection_2021}, who call this group $\Pin^{\tilde c+}(d)$. We will follow Freed-Hopkins' notation, as we will also need $\Pin^{\tilde c-}(d)\coloneqq (\Pin^-(d)\ltimes \U(1))/\Z_2^F$.
\end{example}

The other seven classes in the tenfold way can be worked out in a similar manner. We summarize the results of each step in \cref{tenfold_table}.\footnote{These are not the only conventions for the superalgebras, fermionic groups, and spacetime tangential structures in the literature: see~\cite{stehouwer_interacting_2022} and the references therein.}

\begin{table}[h!]
\begin{tabular}{r c c c r@{\null\;=\;\null}l}
\toprule
$s$ & AZ class & $A$ & $S(A)$ & $H^c(s)$ & $H(S(A))$\\
\midrule
$0$ & A & $\C$ & $\U(1)$ & $H^c(0)$ & $\Spin^c$\\
$1$ & AIII & $\Cl_1\otimes\C$ & $\U(1)\times\Z_2$ & $H^c(1)$ & $\Pin^c$\\
\bottomrule
\end{tabular}
\\[0.7cm]
\begin{tabular}{r c c c r@{\null\;=\;\null}l}
\toprule
$s$ & AZ class & $A$ & $S(A)$ & $H(s)$ & $H(S(A))$\\
\midrule
$-3$ & CII & $\Cl_{-3}$ & $\Pin^-(3)$ & $H(-3)$ & $\Pin^{h-}\coloneqq \Pin^-\times_{\set{\pm 1}}\SU(2)$\\
$-2$ & AII & $\Cl_{-2}$ & $\Pin^-(2)$ & $H(-2)$ & $\Pin^{\tilde c+}\coloneqq \Pin^+\ltimes_{\set{\pm 1}} \U(1)$\\
$-1$ & DIII & $\Cl_{-1}$ & $\Pin^-(1)$
& $H(-1)$ & $\Pin^+$\\
$0$ & D & $\R$ & $\Spin(1)$ 
& $H(0)$ & $\Spin$\\
$1$ & BDI & $\Cl_1$ & $\Pin^+(1)$
& $H(1)$ & $\Pin^-$\\
$2$ & AI & $\Cl_2$ & $\Pin^+(2)$ & $H(2)$ & $\Pin^{\tilde c-}\coloneqq \Pin^-\ltimes_{\set{\pm 1}} \U(1)$\\
$3$ & CI & $\Cl_3$ & $\Pin^+(3)$ & $H(3)$ & $\Pin^{h+}\coloneqq \Pin^+\times_{\set{\pm 1}}\SU(2)$\\
$4$ & C & $\Cl_4$ & $\Spin(3)$
& $H(4)$ & $\Spin^h\coloneqq \Spin\times_{\set{\pm 1}}\SU(2)$\\
\bottomrule
\end{tabular}
\caption{Summary of the procedure outlined in \S\ref{spacetime_sym} beginning with an Altland-Zirnbauer class (second column) and then building a super division algebra $A$ (third column), a fermionic group $S(A)$ (fourth column), and a tangential structure $H(S(A))$ (fifth column). For the tangential structures, the maps to $\O$ are all trivial on $\U(1)$ and $\SU(2)$ and are the usual maps $\Spin\to\SO\to\O$ or $\Pin^\pm\to\O$ on the other factors; since $\set{\pm 1}$ is in the kernel of all of these maps, these maps descend across the quotient by $\set{\pm 1}$ to produce well-defined maps $H(s)\to\O$. First table: the two complex cases. Second table: the eight real cases. Tables adapted from~\cite[(9.24), (9.25)]{freed_reflection_2021} and~\cite[Table 1]{stehouwer_interacting_2022}.}
\label{tenfold_table}
\end{table}

\subsubsection{What changes for weak phases?}
To the best of our knowledge, a mathematical model for the classification of weak phases with interactions has not been widely applied in the literature.
In this subsubsection, we propose a homotopical model in \cref{WI_corollary} using families of invertible field theories over the spatial torus, building from an ansatz of Freed-Hopkins~\cite[Ansatz 2.1]{freed_invertible_2019}.

\begin{ansatz}
\label{WI_ansatz}
The data of a discrete translation-invariant topological phase is equivalent to a family of phases parametrized by the spatial (unit cell) torus $\mathbf T^d = \mathbb{R}^d/\mathbb{Z}^d$.
\end{ansatz}

Freed and Hopkins propose an ansatz that 
the invertible field theories with (spatial) dimension $d$ on a compact $d$-dimensional manifold $Y$ are classified by (a possibly twisted version of) $\mho_H^{d+2}(Y)$:

\begin{ansatz}[{Freed-Hopkins~\cite[Ansatz 2.1, Remark 2.6]{freed_invertible_2019}}]
\label{FH_spatial_ansatz}
The classification of (interacting) invertible $d$-dimensional phases of symmetry type $\rho\colon H\to\O$ over a compact, stably framed manifold $Y$ is naturally equivalent to the classification of $d$-dimensional reflection-positive IFTs of manifolds with an $H$-structure and a map to $Y$, i.e.\ the generalized cohomology group $\mho_H^{d+2} (Y)$.
\end{ansatz}

On the other hand, phases with $\Z^d$ translation symmetry can be understood from the gauge theory point of view through the so-called crystalline equivalence principle \cite{thorngren_gauging_2018}.
Namely, if $\mathbb{Z}^d$ is a spatial symmetry group and a theory is defined on $\mathbb{R}^d$, there is a procedure for gauging the spatial symmetry and considering the emergent gauge theory on the quotient space $\mathbb{R}^d/ \mathbb{Z}^d$. 
Since Freed and Hopkins also consider their ansatz for stacks~\cite[Ansatz 3.3]{freed_invertible_2019}, we can consider invertible field theories on any quotient $\mathbb{R}^d/G$ with $G$ locally compact. 
See in particular~\cite[Example 2.3]{freed_invertible_2019} where this case is studied for $s=0$, $d=2$.

Freed-Hopkins' ansatz is more general than ours; we include only the special case we need.
\begin{corollary}
\label{WI_corollary}
Assuming \cref{WI_ansatz,FH_spatial_ansatz}, deformation classes of invertible discrete translation-invariant topological phases in (spatial) dimension $d$ and Altland-Zirnbauer class $s$ are classified by $d+1$-dimensional reflection-positive IFTs on $H(s)$-manifolds with a map to $\mathbf T^d$, i.e.\ by the generalized cohomology group $\mho_{H(s)}^{d+2}(\mathbf T^d)$.
\end{corollary}

\cref{WI_corollary} is the mathematically rigorous version of the statement we argue for in \ref{kitaev_conjecture} using free fermions and Kitaev's conjecture.

\subsection{Freed-Hopkins' free-to-interacting map for strong phases}\label{freed_hopkins_F2I}

Freed and Hopkins connect the $K$-theoretic classification of free theories to the invertible-field-theoretic classification of interacting theories using a \textit{free-to-interacting map} (\cite{freed_reflection_2021} (9.71)). The kernel of this map comprises, in the second quantized formalism, the theories that are nontrivial under quadratic terms of creation-annihilation operators but that may be trivialized using higher-order terms: a famous example of such a theory is eight copies of the time-reversal symmetric Majorana chain studied by Fidkowski-Kitaev, which can be trivialized under quartic terms \cite{fidkowski_effects_2010}. The cokernel of this map consists of ``interaction-enabled" phases: interacting phases that have no free analog. For example, there is a class CI superconductor in $d=3$ with an intrinsically interacting phase generating a $\Z_2$ interaction-enabled classification \cite[\S V.B]{wang_interacting_2014}.
Thus, assuming the low-energy TQFT ansatz, the free-to-interacting map allows one to mathematically study the physical questions of whether a free phase is robust to interactions and whether new phases arise in the interacting setting.

The free-to-interacting map is built out of two main ingredients.
The first one, the Atiyah-Bott-Shapiro (ABS) orientation, provides a way to get from a bordism class to a $K$-theory class by taking a twisted index. Then, bordism is Anderson-dual to the interacting IFT classification (recall \cref{anderson_dual}), so to land in IFTs instead of bordism we implement this duality and use the Anderson self-duality of $\KO$-theory.

\subsubsection{ABS Orientation}\label{subsec:ABS}
We start with the ABS map in the real case.
There is a classical ABS map from spin bordism $\Omega^\text{Spin}_*$ to the $\KO$-theory of a point, first defined in \cite[\S 11]{atiyah_bott_1968}. Here, we follow Freed-Hopkins~\cite[\S 9.6.3]{freed_reflection_2021}, who use a model for $\KO$-theory developed in \cite{atiyah_index_1969} and follow~\cite[\S II.7]{lawson_spin_1989}. 
An element in $\Omega_n^\Spin$ is represented by an $n$-dimensional spin manifold, while an element in $\KO_n(\pt)$
is (the equivalence class of) a $\Cl_{n}$-module equipped with a Clifford-linear Fredholm operator.
Choose a spin manifold $M$ with a Riemannian metric $g$, and let $\nabla$ be the induced Levi-Civita connection on the \term{Dirac bundle}
\begin{equation}
    \mathcal{S}\coloneqq P_\text{Spin}\times_{\Spin(n)}\Cl_n,
\end{equation}
where $P_{\text{Spin}}$ is the $\Spin(n)$-principal bundle associated to the spin structure on $M$. 
We obtain a Clifford-linear Dirac operator $\slashed{D}_M\colon C^\infty(\mathcal{S})\to C^\infty(\mathcal{S})$ by acting by the covariant derivative followed by Clifford multiplication
$c\colon TM\times \mathcal{S}\to \mathcal{S}$: 
\begin{equation}
    \mathcal{S} \xrightarrow{\nabla} T^*M \otimes \mathcal{S} \overset{g}{\simeq} TM \otimes \mathcal{S} \hookrightarrow \Cl(TM)\otimes \mathcal{S} \xrightarrow{} \mathcal{S}.
\end{equation}
This operator extends to an operator on the appropriate Sobolev completion $\overline{C^\infty(\mathcal S)}$ of $C^\infty(\mathcal S)$. In local coordinates, for $s\in C^\infty(\mathcal{S})$, $\slashed{D}_M$ has the formula
\begin{equation}
    \slashed{D}_M(s) = \sum e_j \cdot \nabla_{e_j}(s).
\end{equation}
The ABS map
\begin{equation}\label{ABS}
\begin{aligned}
    \ABS\colon \Omega_n^\Spin&\to \KO_n \\
    M &\mapsto (\overline{C^\infty(\mathcal{S})}, \slashed{D}_M)
\end{aligned}
\end{equation}
sends a spin manifold $M$ to the Hilbert space $\overline{C^\infty(\mathcal{S})}$ equipped with the Dirac operator. 

Freed and Hopkins \cite[\S 9.2.2]{freed_reflection_2021} develop its twisted generalizations 
\begin{equation}\label{real_twisted_ABS_eqn}
\ABS_s \colon \Omega^{H(s)}_n\to \KO_{n+s}
\end{equation}
by showing that an $n$-manifold $M$ with $H(s)$-structure has a canonical twisted spinor bundle with a twisted $\Cl_{n+s}$-linear Dirac operator.\footnote{A construction of Stolz~\cite[\S 9.3]{stolz_preprint} overlaps with Freed-Hopkins' definition for $s = \pm 1$: the index theory is the same, but Stolz does not turn it into a map of spectra. See also~\cite{Sto88, Zha17, Fre24} for more on index theory on \pinp and \pinm manifolds.}
\begin{example}[Twisted ABS for class AII]\label{ABSAII}
We go into the details of Freed-Hopkins' construction for the case $s = -2$: see~\cite[\S 9.2.2]{freed_reflection_2021} for the proofs of these assertions.

    In class AII, $H(s) = \Pin^{\tilde c +} \coloneqq \Pin^+\ltimes_{\{\pm 1\}}\U(1)$ (\Cref{tenfold_table}). 
    
    An element of $\Omega_n^{\Pin^{\tilde c+}}$ is represented by an $n$-manifold with $\Pin^{\tilde{c}+}(n)$-structure, which is the same as a lift of the classifying map of the tangent bundle $M\xrightarrow{TM}B\O(n)$ to a map $M \to B\Pin^{\tilde c+}(n)$. This gives us a principal $\Pin^{\tilde c+}(n)$-bundle $P_{\Pin^{\tilde c+}(n)}\to M$. 
    The group $\Pin^{\tilde c+}(n)$ embeds into the superalgebra $\Cl_n\otimes \Cl_{-2}$, as follows from \cite[Lemma 9.27]{freed_reflection_2021}.
    We thus have a \emph{twisted Dirac bundle}
    \begin{equation}
        \mathcal S' \coloneqq P_{\Pin^{\tilde c+}(n)}\times_{\Pin^{\tilde c+}(n)}(\Cl_n\otimes \Cl_{-2})\to M.
            \end{equation}

    We can define a Clifford multiplication map
    \begin{equation}
        c\colon TM\otimes \mathcal S'\to \mathcal S'
    \end{equation}
    by using the Clifford multiplication $TM\otimes \Cl(TM)\to \Cl(TM)$ and tensoring with $\Cl_{-2}$.
    Now choose a Riemannian metric on $M$, and choose a connection on the principal $\Pin^{\tilde c+}(n)$-bundle of frames $P_{\Pin^{\tilde c+}(n)}\to M$ whose induced connection on the principal $\O(n)$-bundle of frames is the Levi-Civita connection. This induces a connection on the twisted Dirac bundle.                     
    Now we can define a \emph{twisted Clifford-linear Dirac operator} $\slashed{D}_M = e_i\cdot \nabla_{e_i}$ acting on sections of $\mathcal S'$ by taking the covariant derivative followed by Clifford multiplication. This acts $\Cl_n\otimes \Cl_{-2}$-linearly, so $(\overline{C^\infty(\mathcal S')}, \slashed{D}_{M})$ gives an element of $\KO_{n-2}(\pt)$.

    For example, on the \pincp manifold $\CP^1\times \CP^1$, this twisted Dirac index evaluates to the generator of $\KO_{2}(\pt)$. We prove this in an indirect manner, using a Smith homomorphism, in Appendix~\ref{app:twABS}; it would be interesting to find an index-theoretic proof.
\end{example}
The twisted ABS map and twisted Dirac operators are discussed in full generality for all symmetry classes $H(s)$ in \cite[\S 9.2.2, 9.2.3]{freed_reflection_2021}.

Just as for the real case, there is an ABS orientation landing in complex $K$-theory. The classical map
\begin{equation}\label{spinc_ABS}
    \ABS^c\colon \Omega^{\Spin^c}_n \to K_n
\end{equation}
is from spin$^c$ bordism to the $K$-homology of a point and sends a spin$^c$ manifold $M$ to the complex $\bC\ell_n$-linear Dirac operator acting on smooth sections of the complex Dirac bundle $\mathcal S \coloneqq P_{\Spin^c(n)}\times_{\Spin^c(n)}\bC\ell_n$.

To incorporate the case $H^c(1) = \Pin^c$, Freed-Hopkins in \cite{freed_reflection_2021} (9.44) develop the twisted generalization of this map:
\begin{equation}\label{complex_twisted_ABS_eqn}
    \ABS^c_1\colon \Omega^{\Pin^c}_n \to K_{n+1}.
\end{equation}

\subsubsection{Anderson self-duality of $K$-theory}

The ABS orientation of the previous subsubsection defines a map from twisted spin bordism to $\KO$-homology. However, the invertible field theories modeling interacting phases are classified by \textit{Anderson-dual} twisted spin bordism (\cref{IFT_class}), while free theories are classified by $\KO$-\textit{cohomology} (\cref{neutral_ansatz}).
To reconcile these descriptions, we may apply Anderson duality (\cref{anderson_dual}) and exploit the Anderson self-duality of $\KO$-theory and $K$-theory ~\cite[Theorem
4.16]{And69}.\footnote{Anderson's proof appears in unpublished lecture notes, and it is also discussed in Yosimura~\cite[Theorem 4]{Yos75}. There are several proofs by a variety of different methods; for example, see
Freed-Moore-Segal~\cite[Proposition B.11]{FMS07a}, Heard-Stojanoska~\cite[Theorem 8.1]{heard_k_2014}, Ricka~\cite[Corollary
5.8]{Ric16}, Greenlees-Meier~\cite[Example 1.5]{GM17}, Greenlees-Stojanoska~\cite[Example 6.1]{GS18}, and Hebestreit-Land-Nikolaus~\cite[Example 2.8]{HLN20}.}

\begin{corollary}\label{self_dual_KO}
    There is an isomorphism of cohomology theories
    \begin{equation}
        (I_\Z \KO)^* \cong \KO^{*-4}.
    \end{equation}
\end{corollary}

For complex $K$-theory, there is actually an isomorphism $I_\Z K^* \cong K^*$ with no shift. However, by Bott periodicity, $K^*\cong K^{*+2}$, so we may choose to insert a fourfold shift.
\begin{corollary}\label{self_dual_K}
    There is an isomorphism of cohomology theories
    \begin{equation}
        (I_\Z K)^* \cong K^{*-4}.
    \end{equation}
\end{corollary}

\subsubsection{Free-to-Interacting Maps}

We now have all of the ingredients we need to define the free-to-interacting maps.

\begin{definition}[{Freed-Hopkins~\cite[Conjecture 9.70]{freed_reflection_2021}}]
\label{conjecture:freedhopkinsF2I}
Let $\ABS_s\colon \Omega^{H(s)}_n \to \KO_{n+s}$ be the twisted ABS map \eqref{real_twisted_ABS_eqn}. Applying Anderson duality (\ref{anderson_dual}) gives a map
\begin{subequations}
\begin{equation}\label{IZphi}
     I_\Z\ABS_s\colon I_\Z \KO^{d+s+2}(\text{\rm --}) \to I_{\Z} \Omega_{H(s)}^{d+2}(\text{\rm --})
\end{equation}
of cohomology theories. The right side is by definition $\mho_{H(s)}^{d+2}$. The left side, by Anderson self-duality of $\KO$-theory (\ref{self_dual_KO}), is identified with $\KO^{d+s-2}(\pt)$, so $I_\Z\ABS_s$ is a map of cohomology theories of the form
\begin{equation}
    \FtwoI_s\colon \KO^{d+s-2}(\text{\rm --})\xrightarrow{\cong} I_\Z\KO^{d+s+2}(\text{\rm --})\xrightarrow{I_\Z\ABS_s} I_\Z\Omega_{H(s)}^{d+2}(\text{\rm --}) = \mho_{H(s)}^{d+2}(\text{\rm --}).
\end{equation}
\end{subequations}
The free to interacting map is the composition
\begin{equation}
    \FtwoI_{s,\mathit{strong}}\coloneqq \FtwoI_s(\pt)\colon\KO^{d+s-2} \to \mho_{H(s)}^{d+2}.
\end{equation}
\end{definition}

The complex version of the free-to-interacting map is given by a similar composition, implicit in \cite{freed_reflection_2021}. We define a natural transformation of cohomology theories $\FtwoI_s^c\colon K^{d+s-2}(\text{--})\to\mho_{H^c(s)}^{d+2}(\text{--})$ just as in \cref{conjecture:freedhopkinsF2I}, then evaluate it on $\pt$ to define $\FtwoI_{s,\mathit{strong}}^c$.
\begin{definition}[{Freed-Hopkins~\cite{freed_reflection_2021}}]\label{conjecture:freedhopkinsF2Icpx}
Let $s$ be a complex symmetry type. The free-to-interacting map for theories in spatial dimension $d$  and of symmetry type $s$ is the composition
\begin{equation}
    \FtwoI_{s,\mathit{strong}}^c\coloneqq \FtwoI_s^c(\pt)\colon K^{d+s-2} \xrightarrow{\cong} I_\Z K^{d+s+2} \xrightarrow{I_\Z\ABS_s^c} I_\Z \Omega_{H^c(s)}^{d+2} = \mho_{H^c(s)}^{d+2},
\end{equation}
where the first arrow is the Anderson self-duality of $K$-theory (\ref{self_dual_K})  and the second map is the Anderson dual of the twisted ABS map defined in~\eqref{spinc_ABS}, \eqref{complex_twisted_ABS_eqn}.
\begin{ansatz}[{Freed-Hopkins~\cite[\S 9.2.6]{freed_reflection_2021}}]
\label{ansatz_FH_strong_F2I}
\hfill
\begin{enumerate}
    \item Under the identifications in \cref{neutral_ansatz,IFT_class} identifying the groups of strong free fermion phases, resp.\ reflection positive IFTs in dimension $d$ and real Altland-Zirnbauer class $s$ with $\KO^{d+s-2}$, resp.\ $\mho_{H(s)}^{d+2}$, the homomorphism assigning to a free fermion Hamiltonian its low-energy invertible field theory is $\FtwoI_{s,\mathit{strong}}$.
    \item The above is true mutatis mutandis for a complex Altland-Zirnbauer class $s$ with $H^c(s)$ in place of $H(s)$, $K$ in place of $\KO$, and $\FtwoI_{s,\mathit{strong}}^c$ in place of $\FtwoI_{s,\mathit{strong}}$.
\end{enumerate}
\end{ansatz}

Recall the motivation for free-to-interacting maps given in \S\ref{freed_hopkins_F2I}: knowing these maps allows us to determine both whether a free-fermion SPT phase is stable to interactions and whether there are interaction-enabled phases that one cannot represent using free-fermion models.
\end{definition}

\begin{example}\label{strong_F2I_AII_ex}
    Return to class AII in dimension $d=3$. Let $x\in \KO^{-1}(\pt)\cong \Z_2$ be a free theory with nontrivial Fu-Kane-Mele invariant, which is the generator of $\Z_2$. Such a theory models for example a conducting surface state of the 3d topological insulator BiSb \cite{teo_surface_2008}.
    Its image under the free-to-interacting map is the deformation class of the \pincp topological field theory whose partition function is described by Witten in~\cite[\S 4.7]{Wit16}. In \cref{appendix_thm}, we show that, when evaluated on the generating manifolds $\CP^1\times \CP^1$, $\CP^2$, and $\RP^4$ of $\Omega^{\Pin^{\tilde{c}+}}_4$, this invariant is nontrivial on the first but trivial on the second two.
    
    Since a free theory with the nontrivial Fu-Kane-Mele invariant is sent to a nontrivial interacting theory generating a $\Z_2$ subgroup of the interacting theories, 
    we see that this strong phase is robust to interactions \cite{freed_reflection_2021}.
    That the this invariant survives the addition of interactions was observed in 
    \cite[\S4.7]{Wit16} and 
    \cite{freed_reflection_2021}, and an interacting $\Z_2$ Fu-Kane-Mele index was recently developed in \cite{bachmann_many-body_2024}.

    We have accounted for a $\Z_2$ subgroup of $\mho^5_{\Pin^{\tilde{c}+}}\cong(\Z_2)^3$; the remaining six elements are not in the image of the free-to-interacting map and thus are interaction-enabled phases. There is a generating set of $\mho^5_{\Pin^{\tilde{c}+}}$ given by the Fu-Kane-Mele theory described above, together with two theories whose partition functions are
    \begin{equation}\label{clDW}
        X\longmapsto (-1)^{\int_X w_2(TX)^2},\qquad\qquad
        X \longmapsto (-1)^{\int_X w_1(TX)^4},
    \end{equation}
    detected by $\CP^2$ and $\RP^4$ respectively. These theories are closely related to \term{classical Dijkgraaf-Witten theories}~\cite[\S 1]{FQ93},\footnote{The reason we write ``are closely related to'' instead of ``are'' is a few key differences between these theories and classical Dijkgraaf-Witten theories: in the former, we integrated a mod $2$ characteristic class, and in the latter, one integrates an $\R/\Z$-valued cohomology class $\omega$. Secondly, Dijkgraaf-Witten theory has a background principal bundle for a finite group $G$, and requires $\omega$ to be a characteristic class of $G$-bundles. Integrating $\R/\Z$-cohomology classes requires an orientation, but integrating mod $2$ cohomology classes does not, so the theories in~\eqref{clDW} are defined on any compact $4$-manifold. See~\cite{Deb20, You20, Kim22, GRY24} for more information on unoriented generalizations of Dijkgraaf-Witten theory. Sometimes, theories given by integrating a characteristic class of the tangent bundle are called \term{gravitational theories}.} in that they are given by a classical action which integrates a characteristic class. Unlike the classical actions of Dijkgraaf-Witten theories, the classes $w_1^4$ and $w_2^2$ do not depend on anything stronger that the homotopy type of $X$---in particular, they are independent of the choice of \pincp structure. One could thus think of these theories as ``bosonic;'' frequently such theories are set aside by researchers investigating fermionic SPTs.

\end{example}
\subsubsection{The free-to-interacting map constrains the spectrum of SPT phases}\label{supercohomology_is_wrong}
The ansatz that interacting phases are classified in terms of invertible field theories---and therefore, thanks to \cref{IFT_class}, in terms of bordism---is not the only model for the classification of interacting phases. The purpose of this subsubsection is to point out that the existence of the free-to-interacting map strongly constrains, via one of its basic properties, the possible models for the classification of interacting phases.

Kitaev proposed the ansatz that invertible phases have the structure of a spectrum $E$: that is, $E$ determines a generalized cohomology theory $E^*$, and the classification of $G$-symmetric SPT phases is the (possibly twisted) cohomology group $E^*(BG)$.\footnote{It is predicted that there are two different versions of $E$, one for bosonic phases and the other for fermionic phases. In this paper we focus on the latter.} The model we have followed, which uses invertible field theories, chooses $E^*$ to be the Anderson dual of spin bordism. But there are two more common choices for fermionic phases: \term{restricted supercohomology} $\SH_1$ as introduced by Freed~\cite[\S 1]{Fre08} and Gu-Wen~\cite{GW14}, and \term{extended supercohomology} $\SH_2$ as defined by Kapustin-Thorngren~\cite{KT17} and Wang-Gu~\cite{WG18, Wang_Super}. See also Gaiotto-Johnson-Freyd~\cite[\S 5.3, 5.4]{Gaiotto_SPT}. 
We will not need to know much about these generalized cohomology theories---only that $\SH_1^*(\pt)$ and $\SH_2^*(\pt)$ are concentrated in degrees $0$ through $3$.

Kitaev's argument producing the structure of a spectrum of invertible phases applies equally well for both the free and the interacting classifications, and the argument is compatible with the free-to-interacting map between them. Therefore we hypothesize that Kitaev's conjecture extends: that \emph{the free-to-interacting map refines to a map of spectra}---indeed, this is how Freed-Hopkins~\cite[\S 9.2]{freed_reflection_2021} construct their free-to-interacting maps. This does not constrain the spectrum of SPT phases very much, though: there are nontrivial maps from the $K$-theory spectrum to both restricted and extended supercohomology.

We can do better with one more piece of information: assume there is a procedure on phases of free fermion theories that is analogous to the field-theoretic process of compactification. After taking a continuum limit, one ought to be able to formulate a topological phase of (spatial) dimension $d$ on a closed $d$-manifold $M$, together with some additional structure such as a lattice, a twisted spin structure for fermionic SPT, etc.\footnote{For example, we might to be able to glue a lattice model on $M$ together from a Hamiltonian description on contractible patches.}
By choosing $M$ to be a product $M = N_1\times N_2$, we can compactify on $N_1$ to pass from a $d$-dimensional phase formulated on $M$ to a $(d - \dim(N_1))$-dimensional phase formulated on $N_2$. With some care applied to the tangential structure on $N_1$,\footnote{See Schommer-Pries~\cite[\S 9]{SP18} for a careful general analysis of the tangential structures needed to compactify; we just need a special case addressed by Yamashita-Yonekura~\cite[\S 7.3]{YY23}.} this procedure is expected to define a homomorphism from $d$-dimensional SPT phases to $(d-\dim(N_1))$-dimensional phases, and it has been applied previously in the condensed matter theory literature, e.g.~\cite{HSHH17, Tan17, RL20}. As this procedure can be applied in the same ways to free and to interacting phases, we expect compactification to commute with the free-to-interacting map. Though this may not literally be compactification on free fermion phases, as it is not yet clear whether the process of putting the theory on a general manifold is possible before taking a continuum limit, we expect a homomorphism of this sort to exist for free fermion phases, and we will refer to this homomorphism as compactification.
\begin{ansatz}\label{dimred_ansatz}
\hfill
\begin{description}
    \item[Physics version] The free-to-interacting map commutes with the procedure of compactifying on closed spin manifolds. 
    \item[Math version] The free-to-interacting map is a map of $MT\Spin$-module spectra.
\end{description}
\end{ansatz}
Here $MT\Spin$ is the spectrum whose associated generalized homology theory is spin bordism. The connection between the two versions of \cref{dimred_ansatz} is discussed by Yamashita-Yonekura~\cite[\S 7.3]{YY23}; see also Tachikawa-Yamashita~\cite[\S 2.2.6]{TY23}. Freed-Hopkins' free-to-interacting maps satisfy \cref{dimred_ansatz}~\cite[\S 10]{freed_reflection_2021}.
\begin{proposition}
Assuming \cref{dimred_ansatz}, $\SH_1$ and $\SH_2$ cannot be the spectrum of fermionic phases.
\end{proposition}
\begin{proof}
Let $K$ denote the K3 surface, which is a closed spin $4$-manifold whose spin bordism class generates $\Omega_4^\Spin$~\cite{Mil63}. The $\MTSpin$-module structure on $\KO$ is through the Atiyah-Bott-Shapiro map~\eqref{ABS}; since this is an isomorphism in degrees $0$ through $7$~\cite{ABP67}, compactifying a free fermion theory represented by a class $x\in\KO^m$ on the K3 surface is the same thing as multiplying $x$ by a generator $a$ of $\KO^{-4}$ (and the minus sign is an artifact of the switch from homology to cohomology). In particular, this map is an isomorphism $\KO^4\to\KO^0$.

As noted above, there is no $n$ such that $\SH_i^n(\pt)$ and $\SH_i^{n-4}(\pt)$ are both nonzero, for $i = 1$ or $i = 2$. Therefore in both restricted and extended supercohomology, compactifying on K3 is the zero map.

Suppose $\SH_i$ with either $i = 1$ or $i = 2$ models the spectrum of interacting fermionic SPTs, and let $\tau$ denote the twist of supercohomology over $B\Z_2$ corresponding to Altland-Zirnbauer class DIII. Then compatibility of the free-to-interacting maps with compactification means that the following diagram must commute:
\begin{equation}\label{bad_diagram}
\begin{gathered}
\begin{tikzcd}
	\textcolor{MidnightBlue}{\Z} & {\KO^4} & {\SH_i^{5+\tau}(B\Z_2)} \\
	\textcolor{MidnightBlue}{\Z} & {\KO^0} & {\SH_i^{1+\tau}(B\Z_2)} & {\textcolor{MidnightBlue}{\Z_2}}
	\arrow[MidnightBlue, "\cong"', from=1-1, to=2-1]
	\arrow["\FtwoI", from=1-2, to=1-3]
	\arrow["a", from=1-2, to=2-2]
	\arrow["0", from=1-3, to=2-3]
	\arrow[MidnightBlue, "{\bmod 2}"', curve={height=23pt}, two heads, from=2-1, to=2-4]
	\arrow["\FtwoI", two heads, from=2-2, to=2-3]
\end{tikzcd}\end{gathered}
\end{equation}
For both $\SH_1$ and $\SH_2$, the classification of interacting class DIII phases in (spacetime) dimension $0$ is $\Z_2$: see, e.g., Wang-Gu~\cite[\S VII.E.2.d]{Wang_Super}.
And the free-to-interacting map in dimension $0$ in class DIII is well-known to be nonzero: see~\cite[\S 9.3.1]{freed_reflection_2021} and the references therein. But this is not compatible with the compactification map $\KO^4\to\KO^0$ being an isomorphism and the compactification map on supercohomology vanishing.
\end{proof}
We model interacting phases in class DIII with Anderson-dualized \pinp bordism; therefore in (spacetime) dimension $4$ we have $\Z_{16}$~\cite[\S 2]{Gia73}, in dimension $0$ we have $\Z_2$ (\textit{ibid.}), and the free-to-interacting maps in these dimensions are surjective~\cite[Corollary 9.83]{freed_reflection_2021}. Therefore we learn that compactifying on K3 is the unique surjective map $\Z_{16}\to\Z_2$, which is dual to the fact that the K3 surface represents $8\in\Omega_4^{\Pin^+}\cong\Z_{16}$~\cite[Lemma 5.3]{KT90}.

\section{The ansatz for the weak free-to-interacting map}
\label{s:second_part_section_2}

Having 
{introduced}
the mathematical framework for strong phases in the previous section, our objective here is to propose and analyze an ansatz
for the free-to-interacting map that applies to weak topological phases. 
As we have argued before, the guiding
principle we will use is that weak phases are protected only by a discrete translation symmetry, while being in the trivial phase from the perspective of the continuum limit.

We begin this section by recalling T-duality, which
interchanges the spatial and Brillouin tori, and then explain how the cohomology of the torus
decomposes into contributions from lower-dimensional cells in \cref{tori_splitting-section}. This leads us to \cref{the_weak_F2I_ansatz}: the free-to-interacting
map for weak phases is a composition of T-duality
with the strong free-to-interacting map parameterized over the spatial torus.
We will see in \cref{F2I_james_splitting} that the result is highly computable, as all difficulties are essentially contained in computing the relevant lower-dimensional strong free-to-interacting maps.
This tool will be applied to physically relevant examples in \cref{example_section}.

\subsection{T-duality}
\label{subsec:T-duality}
Whereas the real torus $\mathbf{T}^d$ appears in our \cref{WI_ansatz} for interacting weak phases, non-interacting fermionic topological phases are traditionally formulated over the crystalline momentum space torus $\bT^d$. These tori behave differently, particularly when symmetries are included.
However, T-duality, a construction that originated in string theory \cite{T-duality-Strings}, precisely relates these two tori in a manner that allows us to recast non-interacting results in terms of $\mathbf{T}^d$ and thus to define a free-to-interacting map.
We note that T-duality has been employed many times to treat problems in non-interacting fermionic topological phases \cite{mathai2015t, mathai2016t, Mathai_2016-Tdualitybulkboundaryhigherdim, Hannabuss_2017, gomi2019crystallographic}.

Recall that associated to a $d$-dimensional lattice $\Pi$ are the unit cell or spatial torus $\mathbf T^d \coloneqq \bR^n/\Pi$ and the momentum space torus or Brillouin zone $\bT^d \coloneqq \Hom(\Pi, \U(1))$.
The Brillouin zone has a $\Z_2$-action given by complex conjugation on $\U(1)$.
The Fourier transform between position and momentum space has a $K$-theoretic analog in the T-duality isomorphisms\footnote{In~\eqref{eqn:KSpKQ-Tduality}, $\mathit{KSp}$ is the $K$-theory of quaternionic bundles and $\mathit{KQ}$ is a Real-equivariant version of $\mathit{KSp}$ introduced by Dupont~\cite{Dup69}.}
\begin{subequations}\label{Tdual}
    \begin{align}
        \label{eqn:KOKR-Tduality}
        \mathrm T_\R\colon \KO^\bullet(\mathbf T^d)\xrightarrow{\sim}&\ \KR^{\bullet - d}(\overline{\bT}{}^d) &\text{TRS squares to $1$}\\
        \label{eqn:KSpKQ-Tduality}
        \mathrm T_{\mathbb H}\colon \mathit{KSp}^\bullet (\mathbf T^d)\xrightarrow{\sim} & \ \mathit{KQ}^{\bullet - d}(\overline{\bT}{}^d)&\text{TRS squares to $-1$}\\
        \label{eqn:KU-Tduality}
        \mathrm T_\C\colon K^\bullet(\mathbf T^d)\xrightarrow{\sim} &\ K^{\bullet - d} (\bT^d) & \text{Chern insulators}
    \end{align}
\end{subequations}
which can be defined in terms of a pull-convolve-push construction for topological bundles called the \emph{Fourier-Mukai transform}.
These isomorphisms have been well studied in the condensed matter literature; see for instance \cite{kitaev_periodic_2009, Hannabuss_2017, Mathai_2016-Tdualitybulkboundaryhigherdim}. 
There is also a $C^*$-algebraic approach to this material: see \cite{Rosenberg-realbaumconnes-orientifolds}. 
Here we review the perspective of~\cite{Hannabuss_2017, Mathai_2016-Tdualitybulkboundaryhigherdim}.

The \emph{Poincaré line bundle} $\mathcal L$ is the complex line bundle on $\mathbf T^d\times \overline{\bT}{}^d= \bR^d/\Pi \times \Hom(\Pi, \U(1))$ obtained as the quotient of the trivial bundle $\bC\times \bR^n \times \overline{\bT}{}^d$ by the $\Pi$-action via characters
\begin{equation}
    \pi\cdot (z, v, \chi)\sim (e^{2\pi i \chi(\pi)}z, v+\pi, \chi).
\end{equation}
Bloch waves come from sections of the restrictions of $\mathcal L$ to different momentum cross-sections $\mathbf T^d\times \{\chi\}\subset\mathbf T^d\times \overline{\bT}{}^d$. The T-duality map 
can then be expressed as a pull-push along the correspondence
\begin{equation}
\begin{gathered}
\begin{tikzcd}[ampersand replacement=\&]
        \& \mathbf T^d\times\overline{\bT}{}^d \arrow[dl, "p", swap]\arrow[dr, "\widehat p"]\& \\
        \mathbf T^d \& \& \overline{\bT}{}^d
    \end{tikzcd}
\end{gathered}
\end{equation}
twisted by the Poincaré line bundle
\begin{equation}
    E\mapsto \widehat p_*(p^*E \otimes\mathcal L)
\end{equation}
where the pushforward $\widehat p_*$ is, intuitively, ``integrating out the $\mathbf T^d$ direction'' and thus reduces the dimension by $d$.\footnote{Rigorously, $\widehat p_*$ is fiber integration in the appropriate $K$-theory.}
To recover \eqref{eqn:KOKR-Tduality}, note that in the presence of TRS squaring to $1$, the involution $k\mapsto -k$ lifts to an antilinear action of $T$ on $\mathcal L$, so $p^*E\otimes \mathcal L$ is a Real bundle in $\KR(\mathbf T^d\times \overline{\bT}{}^d)$ and $\widehat p_*$ is the pushforward in $\KR$-cohomology. Equations \eqref{eqn:KSpKQ-Tduality} and \eqref{eqn:KU-Tduality} follow from similar reasoning.\footnote{Another way to see this is that, because $\mathbf T^d$ has trivial $\mathbb Z_2$ action, $\KO(\mathbf T^d)\simeq \KR(\mathbf T^d)$ and $\mathit{KSp}(\mathbf T^d)\simeq \mathit{KQ}(\mathbf T^d)$ so the T-duality maps \eqref{eqn:KOKR-Tduality} and \eqref{eqn:KSpKQ-Tduality} seem like they change the $K$-theory type but in fact they are Fourier-Mukai transforms internal to $\KR, \mathit{KQ}$ respectively, where one of the sides has trivial $\mathbb Z_2$-action and thus reduces to an ordinary nonequivariant $K$-theory.}

\subsection{Splitting the generalized cohomology of tori}\label{tori_splitting-section} 
The generalized cohomology of a torus $\mathbf{T}^d$ has a convenient 
description in terms of the generalized cohomology of spheres (interpreted as cells in a cellular decomposition of $\mathbf{T}^d$), using the fact that for spaces $X, Y$, there is a homotopy equivalence
\begin{equation}
    \Sigma(X\times Y) \simeq \Sigma X \vee \Sigma Y\vee \Sigma(X\wedge Y).
\end{equation}
This formula is sometimes referred to as the \emph{James splitting}.
For instance, $\Sigma \mathbf{T}^2\simeq \Sigma S^1 \vee\Sigma S^1 \vee \Sigma S^2$.
If we iterate this equivalence over the $d$-fold product of circles 
$\mathbf{T}^d \simeq S^1_1\times\ldots\times S^1_d$ and use the suspension isomorphism for generalized cohomology, we get the following identity, where we've labeled the circle factors for notational clarity.
\begin{lemma}[Binomial formula for cohomology of the torus]
    \label{lemma:jamessplitting}
    The generalized cohomology of the $d$-dimensional torus is 
    \begin{equation}
        \widetilde E_0(
        \mathbf{T}^d
        )\cong \bigoplus_{I\subset\{1, \ldots, n\}} \widetilde E_0(S^I) \cong\bigoplus_{n=1}^d \widetilde E_{-n}^{\oplus {d\choose n}},
    \end{equation}
    where we denote by $S^I \coloneqq \wedge_{i\in I}S^1_i$ the $|I|$-dimensional sphere factor indexed by $I\subset \{1, \ldots, n\}$. 
\end{lemma}
The binomial formula interacts nicely with T-duality. Given a spatial torus $\mathbf T^d= S^1_1\times\ldots\times S^1_d$, the Brillouin zone is 
$\overline{\bT}{}^d\simeq \wh S^1_1\times\ldots\times \wh S^1_d$
where $\wh S^1_i$ are the dual circles (1d ``dual tori'') to $S^1_i$.
Then, 
\begin{lemma}[T-duality and the binomial formula]
    \label{lemma:james-and-t}
    Within the binomial formulas $K(\mathbf T^d)\simeq \bigoplus_{I}K(S^I)$ and 
    $K(\overline{\bT}{}^d)
    \simeq \bigoplus_J K(\widehat S{}^J)$, T-duality maps $K(S^I)$ to $K(\widehat S{}^{I^c})$, where $K$ indicates $\KO, \ \KR, \ \mathit{KSp}, \ \mathit{KQ},\ \text{or}\ \mathit{KU}$ as appropriate, and $I^c$ is the complement of $I$. 
\end{lemma}
We refer the reader to \cite[\S 6]{Mathai_2016-Tdualitybulkboundaryhigherdim} for more details.

\subsection{Comparing strong and weak phases}
\label{ss:compare}

In the previous sections we discussed how the generalized cohomology of the torus splits into several summands, some of which are strong, and some of which are weak. 
Here, we would like to emphasize the result obtained in \cref{lemma:james-and-t} that which cell---top or bottom---of the torus is associated to the strong phase depends on which torus we consider.
On the real space torus $\mathbf{T}^d$, strong phases correspond to the summands in $\KO^{d+s-2}(\text{pt})\subset \KO^{d+s-2}(\mathbf{T}^d)$; i.e.\ the summands coming from the point, or 0-cell of the torus.
On the dual torus, which forms the Brillouin zone, the strong phases instead come from the top cell, in the sense that if one crushes all cells except for the top one, the resulting space is $\Z_2$-equivariantly homeomorphic to $\bar S^d$, and the induced pullback map on phases is the inclusion of the strong phases in the group of weak phases.
In summary, the inclusion of strong phases into the total classification including strong and weak phases interacts with T-duality in the way outlined in the following diagram.

\begin{equation}
\begin{gathered}
    \begin{tikzcd}
        \widetilde{\KR}{}^{s-2}(\bar{S}^d) \ar[r,hookrightarrow] & \KR^{s-2}(\overline{\mathbb{T}}{}^d) \\
        \widetilde{\KO}{}^{d+s-2}(S^0) \ar[u,"\mathrm{T}_\R","\cong"']\ar[r,hookrightarrow] & \KO^{d+s-2}(\mathbf{T}^d) \ar[u,"\mathrm{T}_\R","\cong"']
    \end{tikzcd}
\end{gathered}
\end{equation}

Here the top horizontal arrow is induced by a $\Z_2$-equivariant collapse map $(\overline{\mathbb{T}}{}^d)_+ \to \bar S^d$ and the lower by the crush map $(\mathbf{T}^d)_+ \to S^0$.
This, and the analogous complex $K$-theory diagram, commute by Lemma~\ref{lemma:james-and-t}, as these maps pick out the top cell of $\overline{\mathbb{T}}{}^d$ and the bottom cell of $\mathbf{T}^d$.
These inclusions are split and so realize the strong phases as a direct summand of all phases.

\subsection{Kitaev's conjecture for weak SPT phases and T-duality}\label{kitaev_conjecture}

The various proposals for classifying strong interacting SPT phases in $d+1$ dimensions generally satisfy Kitaev's conjecture~\cite{Kit13, Kit15}; namely, for each dimension $d$ there is a space $\mathcal{S}_{d+1}(G)$ of $(d+1)$-dimensional gapped condensed matter systems with symmetry group $G$ such that the SPT phases correspond to connected components of the space, that is, $\mathit{SPT}_{d+1}(G) = \pi_0(\mathcal{S}_{d+1}(G))$, and the spaces $\mathcal{S}_{d+1}(G)$ form an $\Omega$-spectrum in the sense of algebraic topology. That is, for all $d\ge 0$ there is a (weak) homotopy equivalence
\begin{equation}
    \Omega \mathcal{S}_{d+1}(G) \overset\sim\longrightarrow \mathcal{S}_{d}(G).
\end{equation}
Like any spectrum, $\{\mathcal{S}_{d+1}\}_{d\ge 0}$ has an associated generalized cohomology theory $S_G^{d+1}(X) = [X, \mathcal{S}_{d+1}(G)]$, so the degree-$(d+1)$ cohomology of the point classifies $d$-dimensional SPT phases: $S_G^{d+1}(\mathrm{pt}) = \pi_0(\mathcal{S}_{d+1}(G)) = \mathit{SPT}_{d+1}(G)$.


According to Kitaev~\cite{kitaev_periodic_2009} (see also Freed-Moore~\cite{freed_twisted_2013}),
the classification for 
free tenfold way band insulators
of type $H(s)$ in $d+1$ dimensions is given by $\KR$-theory of the Brillouin zone torus (with its conjugation action) \cite{kitaev_periodic_2009,freed_twisted_2013}:
\begin{equation}
    \mathit{FF}_{d+1}(\mathbb{Z}^d,H(s)) \cong \KR^{s-2}(\overline{\mathbb{T}}{}^d).
\end{equation}
While this does classify phases using a generalized cohomology theory, it does not obviously fit the form required by Kitaev's conjecture since the degree of the $K$-theory group, $s-2$, does not depend at all on the dimension of the system. However, we can rephrase using the T-duality isomorphism \eqref{eqn:KOKR-Tduality} to get
\begin{equation}
    \KR^{s-2}(\overline{\mathbb{T}}{}^d) \cong \KO^{d+s-2}(\mathbf T^d),
\end{equation}
which 
is more clearly an instance of Kitaev's conjecture, with
the spectrum in question being $\Sigma^{s-3} \KO$. This same spectrum satisfies Kitaev's original formulation of the conjecture for strong free fermion SPT phases by mapping to the bottom cell, i.e.
\begin{equation}
    \mathit{FF}_{d+1}^{\mathit{strong}}(\mathbb{Z}^d,H(s)) \cong \KO^{d+s-2}(\pt).
\end{equation}
Hence T-duality allows us to state the classification of weak free fermion phases in a way compatible with Kitaev's conjecture, and suggests that for interacting weak SPT phases we should also 
start with
the spatial torus rather than the Brillouin zone torus. Indeed, Kitaev's conjecture for strong phases dictates the degree of the cohomology (namely, $d+1$) and the space the cohomology theory is applied to (namely, the point), leaving only a choice of spectrum. In a modest extension of Kitaev's framework to weak phases, we would likewise want to fix the degree (again $d+1$) and space; and then, the example of weak free fermion spaces shows that the fixed space should be the spatial torus, with the strong phases corresponding to the bottom cell (a point) of this torus.

This means that to extend Kitaev's conjecture to include weak SPT phases we must add the following homotopy equivalence statement: let $\mathcal{S}_{d+1}(G,\mathbb{Z}^d)$ represent $d+1$-dimensional gapped condensed matter systems with internal symmetry $G$ and discrete translation symmetry $\mathbb{Z}^d$; then 
\begin{equation}
    \mathcal{S}_{d+1}(G,\mathbb{Z}^d)\simeq \mathrm{Map}(\mathbf{T}^d,\mathcal{S}_{d+1}(G)),
\end{equation}
so that when mapping to the bottom cell, the above equation reduces to Kitaev's original conjecture for strong SPT phases.

As noted above, we adopt Freed and Hopkins's ansatz for strong SPT phases, namely that they are classified by the group $\mho_{H(s)}^{d+2}(\pt)$, and thus our heuristic argument above leads us to propose
that $\mho_{H(s)}^{d+2}(\mathbf T^d)$ classifies weak SPT phases, as stated in \cref{WI_corollary}. Notice that our rephrasing of the weak free fermion classification using T-duality means that we now have Kitaev-compliant 
classifications for both weak free fermion and weak SPT phases, which will allow us to comfortably compare them in the next section.

\begin{remark}
\label{rem:nonexistentremark}
    We can interpret Kitaev's proposal for strong phases on the interacting side field-theoretically as follows, compare \cite[\S 3.2]{Gaiotto_SPT}.
    An element of $\mho_{H(s)}^{d+1}(X)$ is a $d$-space-dimensional invertible field theory of symmetry type $s$ equipped with a background field valued in $X$.
    As a special case, the suspension isomorphism $\mho_{H(s)}^{d+1}(S^k) \cong \mho_{H(s)}^{d+1} \oplus \mho_{H(s)}^{d+1-k}$ can be interpreted as follows.
    Given a $d$-dimensional invertible quantum field theory $Z$ with background field valued in $S^k$, this gives a $(d-k)$-dimensional invertible quantum field theory by sending\footnote{If $N$ is an $H(s)$-manifold, we use the stably framed structure on $S^k$ arising from the standard isomorphism $TS^k\oplus\underline\R\cong\underline\R^{k+1}$ to make $N \times S^k$ into an $H(s)$-manifold.}
    \begin{equation}
    N^{d-k} \mapsto Z(N \times S^{k} \xrightarrow{\mathit{pr}} S^k).
    \end{equation}
    The factor $\mho_{H(s)}^{d+1}$ corresponds to elements of $\mho_{H(s)}^{d+1}(S^k)$ which do not depend on the $S^k$-valued background field.
    
    We can in this way also reinterpret elements of $\mho_{H(s)}^{d+2}(\mathbf{T}^d)$ as $(d+1)$-dimensional field theories with target $\mathbf{T}^d$.
    By taking appropriate cells of $\mathbf{T}^d$, we can interpret the lower-dimensional terms in the binomial formula for the cohomology of $\mathbf T^d$.
    Specifically, if $\mathbf{T}^k \subseteq \mathbf{T}^d$ is a subtorus corresponding to a subset of $\{1, \dots, d\}$ of size $k$, we can define a map $\mho_{H(s)}^{d+2}(\mathbf{T}^d) \to \mho_{H(s)}^{d+2-k}$ as
    \begin{equation}
    N^{d-k} \mapsto Z(N \times \mathbf{T}^k \to \mathbf{T}^d)
    \end{equation}
    where the map to $\mathbf{T}^d$ is induced by the inclusion $\mathbf{T}^k \subseteq \mathbf{T}^d$. 
\end{remark}

\subsection{The ansatz for the free-to-interacting map for weak phases}\label{ansatz_subsection}

Now we have all the ingredients we need to state our main ansatz (\cref{the_weak_F2I_ansatz}): a model for the free-to-interacting map in terms of generalized cohomology.

Fix a (spatial) dimension $d$ and symmetry type $s$. In \cref{neutral_ansatz}, we modeled the group of free weak phases as $\KR^{s-2}(\overline{\mathbb T}{}^d)$ in the real case and $K^{s-2}(\mathbb T^d)$ in the complex case, and in \cref{WI_corollary} we modeled the group of interacting weak phases as $\mho_{H(s)}^{d+2}(\mathbf T^d)$ in the real case and $\mho_{H^c(s)}^{d+2}(\mathbf T^d)$ in the complex case.

For the moment restrict to real symmetry types.
We need to get from the Brillouin torus to the spatial torus, so our first step is to use T-duality~\eqref{eqn:KOKR-Tduality} to get from $\KR^{s-2}(\overline{\mathbb T}{}^d)$ to $\KO^{d+s-2}(\mathbf T^d)$. After that, we simply apply the free-to-interacting map $
\FtwoI_s\colon \KO^{d+s-2}(\mathbf T^d)\to \mho_{H(s)}^{d+2}(\mathbf T^d)$ of \cref{conjecture:freedhopkinsF2I}, evaluated on the spatial torus. For complex symmetry types, the story is completely analogous, using the T-duality isomorphism of~\eqref{eqn:KU-Tduality} and the free-to-interacting map from \cref{conjecture:freedhopkinsF2Icpx}.
\begin{ansatz}\hfill
\label{the_weak_F2I_ansatz}
\begin{subequations}
\begin{enumerate}
    \item Let $x\in\KR^{s-2}(\overline{\mathbb T}{}^d)$ be a discrete translation-invariant free fermion theory
    in $d$ dimensions and of \emph{real} symmetry type $s$. The long-range effective theory of $x$ is given by the image of $x$ under the composition
    \begin{equation}
        \FtwoI_{\mathit{weak}}\colon \KR^{s-2}(\overline{\mathbb T}{}^d) \overset{\mathrm T_\R^{-1}}{\underset{\eqref{eqn:KOKR-Tduality}}{\longrightarrow}} \KO^{d{+}s-2}(\mathbf T^d) \overset{\FtwoI_s}{\underset{\eqref{conjecture:freedhopkinsF2I}}{\longrightarrow}} \mho_{H(s)}^{d+2}(\mathbf T^d).
    \end{equation}
    \item Let $x\in K^{s-2}(\mathbb T^d)$ be a discrete translation-invariant free fermion theory in $d$ dimensions and of \emph{complex} symmetry type $s$. The long-range effective theory of $x$ is given by the image of $x$ under the composition
    \begin{equation}
        \FtwoI_{\mathit{weak}}^c\colon K^{s-2}(\mathbb T^d) \overset{\mathrm T_\C^{-1}}{\underset{\eqref{eqn:KU-Tduality}}{\longrightarrow}} K^{d{+}s-2}(\mathbf T^d) \overset{\FtwoI_s^c}{\underset{\eqref{conjecture:freedhopkinsF2Icpx}}{\longrightarrow}} \mho_{H^c(s)}^{d+2}(\mathbf T^d).
    \end{equation}
\end{enumerate}
\end{subequations}
\end{ansatz}

This ansatz has several consequences for free and interacting weak phases. The following consequence refines the observation that weak phases break up into a direct sum of strong phases, which is well-known in the physics literature (see e.g. \cite[\S 7.3 and Proposition F.8]{Xiong_SPT} and \cite[\S 9]{guo_fermionic_2020}), and applies it to our free-to-interacting map.
 The decomposition can be subtle; T-duality is essential for a clear understanding of this phenomenon.
\begin{lemma}[Weak phases are built from strong phases]\label{F2I_james_splitting}
    Write $\FtwoI^d_{weak}$ for the weak free-to-interacting map in dimension $d$ from \cref{the_weak_F2I_ansatz}, and $\FtwoI^d_{strong}$ for the strong free-to-interacting map in dimension $d$ from \cref{conjecture:freedhopkinsF2I}. 
    We have that
    \begin{equation}
        \FtwoI^d_{\mathit{weak}} = \bigoplus_{k=0}^d \binom{d}{k} \FtwoI^{d-k}_{\mathit{strong}}.
    \end{equation}
    The analogous statement is true for the complex free-to-interacting maps.
\end{lemma}
\begin{proof}
    The binomial splitting \cref{lemma:jamessplitting} of the Brillouin zone is equivariant with respect to the involutions on $\overline{\mathbb T}{}^d$~\cite[Theorem 11.8]{freed_twisted_2013}. Therefore there is a $\Z_2$-equivariant stable equivalence $\overline{\bT}{}^d \simeq_\text{stably} \bigvee_{I\subseteq[d]} \bar S^I$, under which the element $x\in\KR^{s-2}(\overline{\bT}{}^d)$ splits into elements $x_I\in \widetilde{\KR}{}^{s-2}(\bar S^I)=\KO^{s-2+|I|}(\pt)$. Under T-duality, by \cref{lemma:james-and-t}, we get elements $\text{Dual}(x_I) = \bar{x}_{I^c}\in \KO^{s-2 + d - |I^c|} = \KO^{s-2 + |I|}$.
    Each $I$ thus gives us a strong free-to-interacting map
    \begin{equation}
    \FtwoI_{\mathit{strong}}^I \colon \widetilde{\KR}{}^{s-2}(\bar{S}{}^{|I|}) \to \mho_{H(s)}^{|I|+2}.
    \qedhere
    \end{equation}
\end{proof}
As a result, the kernel and cokernel of $\FtwoI^d_{\mathit{weak}}$ can be computed from those of $\FtwoI^k_{\mathit{strong}}$ as $k$ varies from $0$ to the dimension:  
\begin{equation}
\ker \FtwoI^d_{\mathit{weak}} = \bigoplus_{k=0}^d \binom{d}{k} \ker \FtwoI^{d-k}_{\mathit{strong}}\ , \qquad  \coker \FtwoI^d_{\mathit{weak}} = \bigoplus_{k=0}^d \binom{d}{k} \coker \FtwoI^{d-k}_{\mathit{strong}}.
\end{equation}
This corollary makes the statement that weak phases are built from strong phases of lower dimension precise within our framework.
There are two physical consequences from this result.
First, if the $k$th strong phase is robust to interactions, then so is $k$th component of the weak phase, and vice versa. 
Similarly, all interaction-enabled weak phases arise from interaction-enabled strong phases in lower dimensions---in our model, there are no interaction-enabled phases that do not arise from lower-dimensional phenomena. We give more examples in \S\ref{example_section}.

\begin{example}\label{weak_F2I_AII_ex}
    Freed-Hopkins~\cite[Corollary 9.93]{freed_reflection_2021} calculated that in class AII, the strong free-to-interacting map is always injective in low degrees, including up to spatial dimension 3. From \cref{F2I_james_splitting}, we conclude that weak phases of translation-invariant class AII insulators in dimensions up to three are always robust to interactions.
    
    In $d=3$ in particular, the QSH phases associated to the three planar surfaces of the insulator are robust to interactions. 
    Meanwhile, there are two interaction-enabled phases associated to the top-dimensional cell, coming from 
    $\coker\FtwoI^3_{\mathit{strong}}$, as discussed in \cref{strong_F2I_AII_ex}. There is also an interaction-enabled phase coming from the zero cell, encoded in $\coker \FtwoI^0_{\mathit{strong}}$; see \cite{wang_classification_2014} for a physics interpretation of this phase.
\end{example}

\begin{remark}
Weak phases are protected by discrete translation symmetries. However, their free and interacting classifications can also be used to study situations in which a certain form of crystalline disorder called a dislocation is present. Dislocations are localized disruptions in the crystalline order that can host topologically protected modes---for example, three-dimensional TIs can host one-dimensional helical modes \cite{ran_one-dimensional_2009}. In \cite{ran_weak_2010}, Ran developed a criterion for when these protected modes could occur: if $\vec{B}$ is the Burgers vector of the dislocation, and $\vec{M}$ is a vector of $(d-1)$-dimensional indices, then helical modes can exist if $\vec{B}\cdot\vec{M}$ takes the possible nonzero value.

Our framework can generalize this condition to the interacting setting and to other symmetry types. Consider the group $\widetilde\mho_H^{d+2}(S^{1}\vee\dotsb\vee S^{1})$ classifying codimension-one weak indices for systems with symmetry type $H$. By the suspension isomorphism and wedge axiom, this group is isomorphic to $(\mho^{d+1}_H)^{\oplus d}$, so we may consider its elements to be $d$-vectors $\vec{M}$ of $d$ spacetime-dimensional invertible field theories. We may still consider the Burgers vector $\vec{B}$ to be a vector in the cubic lattice $\Z^d$. Then the generalized dislocation pairing $\vec{B}\cdot \vec{M}$ is valued in $\mho^{d+1}_H$.
For example, the three-dimensional weak topological insulator and the helical modes condition of \cite{ran_weak_2010} concerns the pairing of $\vec{M}\in (\mho^4_{\Pin^{\tilde c+}})^{\oplus 3} \cong (\Z_2)^3 \subset \mho^5_{\Pin^{\tilde c+}}$ with a vector $\vec{B}\in \Z^3$, which takes a binary value in $\mho^4_{\Pin^{\tilde c+}}\cong \Z_2$.
\end{remark}

\section{Results for 
the tenfold way}\label{example_section}
    
In this section, we apply \cref{F2I_james_splitting} to homotopically compute the 
free-to-interacting map of \cref{the_weak_F2I_ansatz} 
in spatial dimensions $1$, $2$, and $3$ for all symmetry types in the tenfold way. 
Admitting the low-energy IFT ansatz,
our computation yields the groups of free and interacting SPT phases for weak topological insulators and superconductors in these dimensions and identifies unstable\footnote{i.e.\ unstable to interactions} free phases and intrinsically interacting phases.

For illustrative purposes, we will discuss Class AII in detail. 
Class AII includes some of the first weak phases studied in the literature, the weak topological insulators (WTIs) of \cite{fu_topological_2007-1} and \cite{fu_topological_2007-2}. We focus on dimension $3+1$. As reviewed in \cref{AII_d3_Ktheory_section}, free phases of band insulators are given by the group $\KR^{-4}(\overline{\mathbb{T}}{}^3) \cong \Z_2 \oplus \Z_2^3 \oplus \Z$. Under T-duality, this group is isomorphic to the real $\KO$-theory group $\KO^{-1}(\mathbf{T}^3)$ on the spatial torus. Using {the binomial formula} (\cref{lemma:jamessplitting}) we obtain an alternative computation of this group as
\begin{align}\label{KO_AII_calc}
    \KO^{-1}(\mathbf{T}^3) &\cong \textcolor{ceruleanblue}{\KO^{-1}(\pt)} \oplus 3\KO^{-2}(\pt) \oplus 3\KO^{-3}(\pt) \oplus \textcolor{BrickRed}{\KO^{-4}(\pt)} \\
    &= \textcolor{ceruleanblue}{\Z_2}\oplus (\Z_2)^3 \oplus \textcolor{BrickRed}{\Z}.
\end{align}
Here \textcolor{ceruleanblue}{blue} summands are the strong phases, regarded as a subgroup of the group of weak phases (see \S\ref{ss:compare}). The \textcolor{BrickRed}{red} summands are captured by the invariant that counts the number of valence bands. This is generally not an interesting invariant and is often excluded in the physics literature; see \cref{rem:numberofbands} for details. With respect to the binomial formula, strong phases (\textcolor{ceruleanblue}{blue} summands) correspond to the bottom cell of the spatial torus, and the valence-band-counting invariant (\textcolor{BrickRed}{red} summands) corresponds to the top cell.\footnote{As T-duality incorporates Poincaré duality, the description is opposite for the Brillouin torus: the strong phases correspond to the top cell of $\overline{\mathbb T}{}^d$, and the valence-band-counting invariant corresponds to the bottom cell, see \S\ref{ss:compare}.}

In~\eqref{KO_AII_calc}, the strong phase in the $\textcolor{ceruleanblue}{\Z_2}$ summand $\textcolor{ceruleanblue}{\KO^{-1}(\pt)}$ is detected by the Fu-Kane-Mele invariant \cite{fu_topological_2007-2}, while the $\textcolor{BrickRed}{\Z}$ summand counts the number of valence bands.
The remaining $(\Z_2)^3$ coming from $3\KO^{-2}(\pt)$ comprises the weak phases, which may be viewed as quantum spin Hall (QSH) phases localized to two-dimensional surfaces of the three-dimensional material.

We compute the interacting classification using the binomial formula as well. 
We have 
\begin{equation}
\begin{aligned}
    \mho_{\Pin^{\tilde{c}+}}^{5}(\mathbf{T}^3) &\cong \textcolor{BrickRed}{\mho_{\Pin^{\tilde{c}+}}^{5}(\pt)} \oplus 3\mho_{\Pin^{\tilde{c}+}}^{4}(\pt) \oplus 3 \mho_{\Pin^{\tilde{c}+}}^{3}(\pt) \oplus \textcolor{ceruleanblue}{\mho_{\Pin^{\tilde{c}+}}^{2}(\pt)} \\
    &\cong \textcolor{BrickRed}{\Z}\oplus (\Z_2)^3 \oplus \textcolor{ceruleanblue}{(\Z_2)^3}.
\end{aligned}
\end{equation}
Once again, the colors illustrate which phases come from the top and bottom cells of $\mathbf T^d$. The first triplet of $\Z_2$'s comes from the interacting weak phases, while the second triplet of $\Z_2$'s comes from the interacting strong phases in 2d.
That the weak $(\Z_2)^3$ injects into $\mho_{\Pin^{\tilde{c}+}}^{5}(\mathbf{T}^3)$ corroborates the expectation that these weak phases are stable under interactions \cite[\S III.A]{zou_bulk_2018}, \cite{liu_half_2012}.

In $d=3$, there is also a $(\Z_2)^2$ classification of interaction-enabled phases. These phases are all strong; i.e.\ they arise from $\mho_{\Pin^{\tilde c+}}^5$ applied to the bottom cell of the spatial torus. 
These interaction-enabled phases were originally found in the physics literature in \cite{wang_classification_2014} and connected to bordism theory in \cite{Met15}.

\Cref{tab:weakTQ} includes the classification for all three relevant dimensions.

\begin{corollary}[Symmetry class AII, $s = -2$]
\label{tab:weakTQ}
The free-to-interacting map for the groups of weak phases in Altland-Zirnbauer type AII is:
\begin{equation}
\begin{gathered}
\begin{tikzcd}[ampersand replacement = {\&}, row sep=0.2cm, column sep=0.4cm,
    execute at end picture={
        \foreach \y in {1.45, -1.65} {
            \draw[thick] (-5.3, \y) -- (5.3, \y);
        }
        \draw (-5.3, 0.85) -- (5.3, 0.85);
    }
]
	d \& {\ker(\FtwoI)} \& {\KO^{d-4}(\mathbf T^d)} \& {\mho^{d+2}_{\Pin^{\tilde c+}}(\mathbf T^d)} \& {\coker(\FtwoI)} \\
	1 \& 0 \& \textcolor{BrickRed}{\Z} \& \textcolor{BrickRed}{\Z} \& 0 \\
	2 \& 0 \& {\textcolor{BrickRed}{\Z}\oplus \textcolor{ceruleanblue}{\Z_2}} \& {\textcolor{BrickRed}{\Z} \oplus \textcolor{ceruleanblue}{\Z_2}} \& 0 \\
	3 \& 0 \& {\textcolor{BrickRed}{\Z} \oplus \Z_2^3 \oplus \textcolor{ceruleanblue}{\Z_2}} \& {\textcolor{BrickRed}{\Z} \oplus \Z_2^3 \oplus \textcolor{ceruleanblue}{\Z_2^3}} \& {\textcolor{ceruleanblue}{\Z_2^2}}
	\arrow[from=1-2, to=1-3]
	\arrow["\FtwoI", from=1-3, to=1-4]
	\arrow[from=1-4, to=1-5]
\end{tikzcd}\end{gathered}
\end{equation}
\end{corollary}
\begin{litnote}
The classification of these free weak phases has been studied from many perspectives in the literature: see, for example, De Nittis-Gomi~\cite{DNG15, DNG18b, DNG18a, DG23a, DNG23}, Fiorenza-Monaco-Panati~\cite{FMP16}, and Kaufmann-Li-Wehefritz-Kaufmann \cite{KLW16, KLWK16, kaufmann2016stiefelwhitney, KLWK24}.
Pin\textsuperscript{$\tilde c+$} bordism groups in these dimensions were first computed by Freed-Hopkins~\cite[Theorem 9.87]{freed_reflection_2021}.
\end{litnote}

We continue with the seven other real symmetry types and the two complex symmetry types. 
\begin{corollary}[Symmetry class D, $s = 0$]
\label{tab:weak_real_s=0}
The free-to-interacting map for the groups of weak phases in Altland-Zirnbauer type D is:
\begin{equation}
\begin{gathered}
\begin{tikzcd}[ampersand replacement = {\&}, row sep=0.2cm, column sep=0.4cm,
    execute at end picture={
        \foreach \y in {1.55, -1.7} {
            \draw[thick] (-5.4, \y) -- (5.4, \y);
        }
        \draw (-5.4, 0.9) -- (5.4, 0.9);
    }
]
	d \& {\ker(\FtwoI)} \& {\KO^{d-2}(\mathbf T^d)} \& {\mho^{d+2}_{\Spin}(\mathbf T^d)} \& {\coker(\FtwoI)} \\
	1 \& 0 \& \textcolor{ceruleanblue}{\Z_2} \oplus\textcolor{BrickRed}{\Z_2} \& \textcolor{ceruleanblue}{\Z_2} \oplus\textcolor{BrickRed}{\Z_2} \& 0 \\
	2 \& 0 \& \textcolor{ceruleanblue}{\Z}\oplus \textcolor{BrickRed}{\Z_2} \oplus \Z_2^2 \& \textcolor{ceruleanblue}{\Z}\oplus \textcolor{BrickRed}{\Z_2} \oplus \Z_2^2 \& 0 \\
	3 \& 0 \& \Z^3 \oplus\textcolor{BrickRed}{\Z_2} \oplus  \Z_2^3 \& \Z^3 \oplus \textcolor{BrickRed}{\Z_2} \oplus \Z_2^3 \& 0
	\arrow[from=1-2, to=1-3]
	\arrow["\FtwoI", from=1-3, to=1-4]
	\arrow[from=1-4, to=1-5]
\end{tikzcd}\end{gathered}
\end{equation}
\end{corollary}
\begin{litnote}
Hughes~\cite[6m29s]{hughes_weak_2015} gives a classification of the weak free phases in these dimensions, modulo the band-counting $\textcolor{BrickRed}{\Z_2}$ subgroup. Freed-Hopkins~\cite[Example 2.3]{freed_invertible_2019} study class D phases on a torus in dimension $2$, and observe that the free-to-interacting map is an isomorphism in that dimension; this example is also studied by Ran~\cite{ran_weak_2010} and Rao-Sodeman~\cite[\S IV.C]{RS21}, whose results also agree with ours.
The spin bordism groups used in \cref{tab:weak_real_s=0} were first computed by Milnor~\cite{Mil63}.
\end{litnote}
\begin{corollary}[Symmetry class BDI, $s = 1$]
\label{tab:weak_real_s=1}
The free-to-interacting map for the groups of weak phases in Altland-Zirnbauer type BDI is:
\begin{equation}
\begin{gathered}
\begin{tikzcd}[ampersand replacement = {\&}, row sep=0.2cm, column sep=0.4cm,
    execute at end picture={
        \foreach \y in {1.55, -1.7} {
            \draw[thick] (-5.05, \y) -- (5.05, \y);
        }
        \draw (-5.05, 0.9) -- (5.05, 0.9);
    }
]
	d \& {\ker(\FtwoI)} \& {\KO^{d-1}(\mathbf T^d)} \& {\mho^{d+2}_{\Pin^-}(\mathbf T^d)} \& {\coker(\FtwoI)} \\
	1 \& \textcolor{ceruleanblue}{8\Z} \& \textcolor{ceruleanblue}{\Z} \oplus \textcolor{BrickRed}{\Z_2} \& \textcolor{ceruleanblue}{\Z_8} \oplus \textcolor{BrickRed}{\Z_2} \& 0 \\
	2 \& (8\Z)^2 \& \Z^2 \oplus \textcolor{BrickRed}{\Z_2} \& \Z_8^2 \oplus \textcolor{BrickRed}{\Z_2} \& 0 \\
	3 \& (8\Z)^3 \& \Z^3 \oplus \textcolor{BrickRed}{\Z_2} \& \Z_8^3\oplus \textcolor{BrickRed}{\Z_2} \& 0
	\arrow[from=1-2, to=1-3]
	\arrow["\FtwoI", from=1-3, to=1-4]
	\arrow[from=1-4, to=1-5]
\end{tikzcd}\end{gathered}
\end{equation}
\end{corollary}
\begin{litnote}
\Cref{tab:weak_real_s=1} is in agreement with work of Xiao-Kawabata-Luo-Ohtsuki-Shindou~\cite{xiao_anisotropic_2023}, who study 3d class BDI weak phases and conclude that the three $\Z$-valued invariants of weak topological phases remain nontrivial in the presence of interactions.

The \pinm bordism groups used in this computation were first computed by Anderson-Brown-Peterson~\cite{ABP69}. The Majorana chain with its time-reversal symmetry is a $1$-dimensional strong phase in class BDI, generating the $\textcolor{ceruleanblue}{\Z}$ summand of free phases and the $\textcolor{ceruleanblue}{\Z_8}$ summand of interacting phases in $d = 1$~\cite{Kit01, fidkowski_effects_2010, FK11, turner_topological_2011}; this phase thus defines higher-dimensional weak phases and so contributes to the kernel of the free-to-interacting map in all higher degrees. 
We predict no interaction-enabled phases in dimensions $6$ and below.
\end{litnote}
\begin{corollary}[Symmetry class AI, $s = 2$]
\label{tab:weak_real_s=2}
The free-to-interacting map for the groups of weak phases in Altland-Zirnbauer type AI is:
\begin{equation}
\begin{gathered}
\begin{tikzcd}[ampersand replacement = {\&}, row sep=0.2cm, column sep=0.4cm,
    execute at end picture={
        \foreach \y in {1.5, -1.7} {
            \draw[thick] (-4.9, \y) -- (4.9, \y);
        }
        \draw (-4.9, 0.9) -- (4.9, 0.9);
    }
]
	d \& {\ker(\FtwoI)} \& {\KO^{d}(\mathbf T^d)} \& {\mho^{d+2}_{\Pin^{\tilde c-}}(\mathbf T^d)} \& {\coker(\FtwoI)} \\
	1 \& 0 \& \textcolor{BrickRed}{\Z}  \& \textcolor{BrickRed}{\Z}\oplus \Z_2  \& \Z_2 \\
	2 \& 0 \& \textcolor{BrickRed}{\Z}  \& \textcolor{BrickRed}{\Z}\oplus \Z_2^2 \& \Z_2^2 \\
	3 \& 0 \& \textcolor{BrickRed}{\Z} \& \textcolor{BrickRed}{\Z} \oplus \Z_2^4 \& \Z_2^4
	\arrow[from=1-2, to=1-3]
	\arrow["\FtwoI", from=1-3, to=1-4]
	\arrow[from=1-4, to=1-5]
\end{tikzcd}\end{gathered}
\end{equation}
\end{corollary}
\begin{litnote}
The \pincm bordism groups used in this computation were computed by Freed-Hopkins~\cite[Theorem 9.87]{freed_reflection_2021}.
Like in \cref{tab:weak_real_s=1}, the interaction-enabled weak phases in dimensions $2$ and $3$ are a consequence of the interaction-enabled \emph{strong} phase in dimension $1$ in this class; this strong phase appears in Freed-Hopkins~\cite[Corollary 9.95]{freed_reflection_2021} (they use spacetime dimension, so call that phase $2$-dimensional). De Nittis-Gomi~\cite{DNG14, DNG16} classify the free weak phases in class AI using Real-equivariant vector bundles.
\end{litnote}
\begin{corollary}[Symmetry class CI, $s = 3$]
\label{tab:weak_real_s=3}
The free-to-interacting map for the groups of weak phases in Altland-Zirnbauer type CI is:
\begin{equation}
\begin{gathered}
\begin{tikzcd}[ampersand replacement = {\&}, row sep=0.2cm, column sep=0.4cm,
    execute at end picture={
        \foreach \y in {1.55, -1.7} {
            \draw[thick] (-5.25, \y) -- (5.25, \y);
        }
        \draw (-5.25, 0.9) -- (5.25, 0.9);
    }
]
	d \& {\ker(\FtwoI)} \& {\KO^{d+1}(\mathbf T^d)} \& {\mho^{d+2}_{\Pin^{h+}}(\mathbf T^d)} \& {\coker(\FtwoI)} \\
	1 \& 0 \& 0  \& \textcolor{ceruleanblue}{\Z_2}  \& \textcolor{ceruleanblue}{\Z_2} \\
	2 \& 0 \& 0  \&  \Z_2^2 \& \Z_2^2 \\
	3 \& \textcolor{ceruleanblue}{4\Z} \& \textcolor{ceruleanblue}{\Z} \& \textcolor{ceruleanblue}{\Z_4\oplus  \Z_2}\oplus \Z_2^3 \& \textcolor{ceruleanblue}{\Z_2}\oplus \Z_2^3
	\arrow[from=1-2, to=1-3]
	\arrow["\FtwoI", from=1-3, to=1-4]
	\arrow[from=1-4, to=1-5]
\end{tikzcd}\end{gathered}
\end{equation}
\end{corollary}
\begin{litnote}
The \pinhp bordism groups appearing in this computation were computed by Freed-Hopkins~\cite[Theorem 9.97]{freed_reflection_2021}; they use the notation $G^+$ for $\Pin^{h+}$.
Just as in \cref{tab:weak_real_s=2}, the interaction-enabled strong phase in dimension $1$, first studied by Freed-Hopkins~\cite[Corollary 9.101]{freed_reflection_2021}, gives rise to interaction-enabled weak phases in higher dimensions.
\end{litnote}
\begin{corollary}[Symmetry class C, $s = 4$]
\label{tab:weak_real_s=4}
The free-to-interacting map for the groups of weak phases in Altland-Zirnbauer type C is:
\begin{equation}
\begin{gathered}
\begin{tikzcd}[ampersand replacement = {\&}, row sep=0.2cm, column sep=0.4cm,
    execute at end picture={
        \foreach \y in {1.5, -1.6} {
            \draw[thick] (-5.1, \y) -- (5.1, \y);
        }
        \draw (-5.1, 0.8) -- (5.1, 0.8);
    }
]
	d \& {\ker(\FtwoI)} \& {\KO^{d+2}(\mathbf T^d)} \& {\mho^{d+2}_{\Spin^{h}}(\mathbf T^d)} \& {\coker(\FtwoI)} \\
	1 \& 0 \& 0  \& 0  \& 0 \\
	2 \& 0 \& \textcolor{ceruleanblue}{\Z}  \&  \textcolor{ceruleanblue}{\Z^2} \& \textcolor{ceruleanblue}{\Z} \\
	3 \& 0 \& \Z^3 \& \Z^6 \& \Z^3
	\arrow[from=1-2, to=1-3]
	\arrow["\FtwoI", from=1-3, to=1-4]
	\arrow[from=1-4, to=1-5]
\end{tikzcd}\end{gathered}
\end{equation}
\end{corollary}
\begin{litnote}
Both a free weak phase and an interaction-enabled weak phase contribute to the classification in $d=3$. The \spinh bordism groups used in the computation in \cref{tab:weak_real_s=4} were first computed by Freed-Hopkins~\cite[Theorem 9.97]{freed_reflection_2021}, though they use the notation $G^0$ for $\Spin^h$.
\end{litnote}
\begin{corollary}[Symmetry class CII, $s = -3$]
\label{tab:weak_real_s=-3}
The free-to-interacting map for the groups of weak phases in Altland-Zirnbauer type CII is:
\begin{equation}
\begin{gathered}
\begin{tikzcd}[ampersand replacement = {\&}, row sep=0.2cm, column sep=0.4cm,
    execute at end picture={
        \foreach \y in {1.6, -1.75} {
            \draw[thick] (-5.15, \y) -- (5.15, \y);
        }
        \draw (-5.15, 0.9) -- (5.15, 0.9);
    }
]
	d \& {\ker(\FtwoI)} \& {\KO^{d+3}(\mathbf T^d)} \& {\mho^{d+2}_{\Pin^{h-}}(\mathbf T^d)} \& {\coker(\FtwoI)} \\
	1  \& \textcolor{ceruleanblue}{2\Z}  \& \textcolor{ceruleanblue}{\Z}  \& \textcolor{ceruleanblue}{\Z_2} \& 0 \\
	2 \& (2\Z)^2 \& \Z^2  \&  \Z_2^2 \& 0 \\
	3 \& (2\Z)^3 \& \Z^3 \oplus \Z_2  \&  \Z_2^6 \& \Z_2^2
	\arrow[from=1-2, to=1-3]
	\arrow["\FtwoI", from=1-3, to=1-4]
	\arrow[from=1-4, to=1-5]
\end{tikzcd}\end{gathered}
\end{equation}
\end{corollary}
\begin{litnote}
Just as in \cref{tab:weak_real_s=1}, the generator of the group of one-dimensional strong phases becoming torsion in the interacting classification (\cite[Corollary 9.103]{freed_reflection_2021}) gives rise to weak phases in the kernel of the free-to-interacting map in higher dimensions. In $d=3$, we also predict
two interaction-enabled phases. Xiao-Kawabata-Luo-Ohtsuki-Shindou~\cite{xiao_anisotropic_2023} briefly discuss 3d weak interacting phases in class CII: they claim that the three $\Z$-valued topological indices of free phases remain nontrivial in the presence of interactions, which our computations support.

The \pinhm bordism groups that we used in this computation are computed by Freed-Hopkins~\cite[Theorem 9.97]{freed_reflection_2021}; they write $G^-$ for $\Pin^{h-}$.
\end{litnote}
\begin{corollary}[Symmetry class DIII, $s = -1$]
\label{tab:weak_real_s=-1}
The free-to-interacting map for the groups of weak phases in Altland-Zirnbauer type DIII is:
\begin{equation}
\begin{gathered}
\begin{tikzcd}[ampersand replacement = {\&}, row sep=0.2cm, column sep=0.4cm,
    execute at end picture={
        \foreach \y in {1.6, -1.75} {
            \draw[thick] (-5.1, \y) -- (5.1, \y);
        }
        \draw (-5.1, 0.9) -- (5.1, 0.9);
    }
]
	d \& {\ker(\FtwoI)} \& {\KO^{d-3}(\mathbf T^d)} \& {\mho^{d+2}_{\Pin^+}(\mathbf T^d)} \& {\coker(\FtwoI)} \\
	1 \& 0 \& \textcolor{ceruleanblue}{\Z_2} \& \textcolor{ceruleanblue}{\Z_2} \& 0 \\
	2 \& 0 \& \Z_2^2 \oplus \textcolor{ceruleanblue}{\Z_2} \& \Z_2^2 \oplus \textcolor{ceruleanblue}{\Z_2} \& 0 \\
	3 \& \textcolor{ceruleanblue}{16\Z} \& \textcolor{ceruleanblue}{\Z}\oplus \Z_2^6 \& \textcolor{ceruleanblue}{\Z_{16}}\oplus \Z_2^6 \& 0
	\arrow[from=1-2, to=1-3]
	\arrow["\FtwoI", from=1-3, to=1-4]
	\arrow[from=1-4, to=1-5]
\end{tikzcd}\end{gathered}
\end{equation}
\end{corollary}
\begin{litnote}
The \pinp bordism groups used in this computation were computed by Giambalvo~\cite[\S 2]{Gia73}. De Nittis-Gomi~\cite{de_nittis_cohomology_2022} classify the free weak phases in this class using equivariant cohomology. The weak phases arising from $d=1$ are stable to interactions. The strong phase in $d=3$ breaks from generating a $\Z$ of free phases to a $\Z_{16}$ of interacting phases; this has been argued in many different ways: see for example~\cite{Kit11, FCV13,
MFCV14, wang_interacting_2014, YX14, Kit15, Kapustin_SPT, TY16, Wit16, SSR17a, Wang_Super, freed_reflection_2021}.
\end{litnote}
\begin{corollary}[Symmetry class A, $s = 0$]
\label{tab:weakA}
The free-to-interacting map for the groups of weak phases in Altland-Zirnbauer type A is:
\begin{equation}
\begin{gathered}
\begin{tikzcd}[ampersand replacement = {\&}, row sep=0.2cm, column sep=0.4cm,
    execute at end picture={
        \foreach \y in {1.55, -1.75} {
            \draw[thick] (-4.8, \y) -- (4.8, \y);
        }
        \draw (-4.8, 0.85) -- (4.8, 0.85);
    }
]
	d \& {\ker(\FtwoI)} \& {K^d(\mathbf T^d)} \& {\mho^{d+2}_{\Spin^c}(\mathbf T^d)} \& {\coker(\FtwoI)} \\
	1 \& 0 \& \textcolor{BrickRed}{\Z} \& \textcolor{BrickRed}{\Z} \& 0 \\
	2 \& 0 \& \textcolor{BrickRed}{\Z} \oplus \textcolor{ceruleanblue}{\Z} \& \textcolor{BrickRed}{\Z} \oplus \textcolor{ceruleanblue}{\Z^2} \& \textcolor{ceruleanblue}{\Z} \\
	3 \& 0 \& \textcolor{BrickRed}{\Z}\oplus \Z^3 \& \textcolor{BrickRed}{\Z}\oplus \Z^6 \& \Z^3
	\arrow[from=1-2, to=1-3]
	\arrow["\FtwoI", from=1-3, to=1-4]
	\arrow[from=1-4, to=1-5]
\end{tikzcd}\end{gathered}
\end{equation}
\end{corollary}
\begin{litnote}
Varjas-de Juan-Lu~\cite[\S II]{varjas_space_2017} observe that the Hall conductivity, a $\Z$-valued invariant of weak free class A phases in 3d, remains a well-defined, $\Z$-valued invariant of interacting systems, which is consistent with our computations. The $\Z$ summand in $d=2$ corresponding to the integer quantum Hall effect (a strong phase) is stable under interactions and contributes to weak phases in $d=3$. There are also interaction-enabled phases in $d=2$, which contribute to weak interaction-enabled phases in the $d=3$ Chern insulator. The calculation of \spinc bordism groups is attributed to Anderson-Brown-Peterson~\cite{ABP67}; see Bahri-Gilkey~\cite[\S 1]{BG87a} for an explicit description.
\end{litnote}
\begin{corollary}[Symmetry class AIII, $s = 1$]
\label{tab:weakAIII}
The free-to-interacting map for the groups of weak phases in Altland-Zirnbauer type AIII is:
\begin{equation}
\begin{gathered}
\begin{tikzcd}[ampersand replacement = {\&}, row sep=0.2cm, column sep=0.4cm,
    execute at end picture={
        \foreach \y in {1.65, -1.65} {
            \draw[thick] (-5.35, \y) -- (5.35, \y);
        }
        \draw (-5.35, 0.9) -- (5.35, 0.9);
    }
]
	d \& {\ker(\FtwoI)} \& {K^{d-1}(\mathbf T^d)} \& {\mho^{d+2}_{\Pin^c}(\mathbf T^d)} \& {\coker(\FtwoI)} \\
	1 \& \textcolor{ceruleanblue}{4\Z}\& \textcolor{ceruleanblue}{\Z} \& \textcolor{ceruleanblue}{\Z_4} \& 0 \\
	2 \& (4\Z)^2 \& \Z^2 \& \Z_4^2 \& 0 \\
	3 \& \textcolor{ceruleanblue}{8\Z} \oplus (4\Z)^3 \& \textcolor{ceruleanblue}{\Z}\oplus \Z^3 \& \textcolor{ceruleanblue}{\Z_8\oplus\Z_2} \oplus \Z_4^3 \& \textcolor{ceruleanblue}{\Z_2}
	\arrow[from=1-2, to=1-3]
	\arrow["\FtwoI", from=1-3, to=1-4]
	\arrow[from=1-4, to=1-5]
\end{tikzcd}\end{gathered}
\end{equation}
\end{corollary}
\begin{litnote}
\Cref{tab:weakAIII} is in agreement with work of Xiao-Kawabata-Luo-Ohtsuki-Shindou~\cite{xiao_anisotropic_2023}, who discuss how the weak $\Z$-valued indices in 3d remain nontrivial in the presence of interactions. See also Claes-Hughes~\cite{claes_disorder_2020}, who study the behavior of these indices under disorder.
De Nittis-Gomi~\cite{DNG18c} study free weak phases in this class in terms of objects called \term{chiral vector bundles}.

Like in \cref{tab:weak_real_s=1,tab:weak_real_s=-3}, the nontrivial kernel of the one-dimensional strong free-to-interacting map (\cite[Corollary 9.91]{freed_reflection_2021}) produces phases in the kernel of the weak free-to-interacting map in higher dimensions. The generator of the group of one-dimensional strong phases is the class of the Su-Schrieffer-Heeger model~\cite{SSH79}. In three dimensions 
we predict
an additional free phase that breaks down, as well as an interaction-enabled phase.  Pin\textsuperscript{$c$} bordism groups were first computed by Bahri-Gilkey~\cite{BG87a, BG87b}.
\end{litnote}

\subsection*{Statement on Conflicts of Interest} The authors have no competing interests to declare that
are relevant to the content of this article.

\appendix

\section{Calculation of the twisted ABS map \texorpdfstring{$\Omega_4^{\Pin^{\tilde c+}}\to\Z_2$}{}}
    \label{app:twABS}
Our goal in this appendix is to explicitly calculate Freed-Hopkins' twisted Atiyah-Bott-Shapiro map in dimension $4$ in class $s = -2$. We use this calculation in \cref{strong_F2I_AII_ex}.

Freed-Hopkins' original calculation of this free-to-interacting map in~\cite[\S 10]{freed_reflection_2021} is purely homotopy-theoretical, coming from an Adams spectral sequence computation. We make a more concrete and less technical calculation, which additionally results in an explicit understanding of the manifold generators of the relevant bordism groups.
Specifically, we use the long exact sequence associated to a \term{Smith homomorphism}, following a general theory worked out in~\cite{mathSmith}. See~\cite{HS13, DL23, DDHM23, Deb24, DNT24, DK24, DYY23} for more examples of calculations applying this technique and~\cite[Lemma 5.19]{miyazawa_localization_2025} for a similar calculation of Freed-Hopkins' map in dimension $3$.
\begin{remark}
Freed-Hopkins' original definition of the twisted Atiyah-Bott-Shapiro map is index-theoretic, as we review in \S\ref{subsec:ABS}. It therefore seems reasonable that there should be a description of the map $\Omega_4^{\Pin^{\tilde c+}}\to\KO_2\cong\Z_2$ as a mod $2$ index of the Dirac operator from \cref{ABSAII}. We would be interested in learning whether it is possible to prove \cref{appendix_thm} by calculating this mod $2$ index on a generating set for $\Omega_4^{\Pin^{\tilde c+}}$.
\end{remark}
\begin{definition}[{Hason-Komargodski-Thorngren~\cite[\S 4.1]{hason_anomaly_2020}}]
Let $V\to X$ be a virtual vector bundle. An \term{$(X, V)$-twisted spin structure} on a vector bundle $E\to M$ is data of a map $f\colon M\to X$ and a spin structure on $E\oplus f^*(V)$.
\end{definition}
Given a fermionic group $G_f$, there is often a virtual vector bundle $V\to BG_b$ such that the tangential structure associated to $G$ as defined in \S\ref{spacetime_sym} is equivalent to a $(BG_b, V)$-twisted spin structure. This occurs for all fermionic groups we consider in this paper; that it is not true in general follows from work of Gunarwardena-Kahn-Thomas~\cite{GKT89}. We will need the following three examples.
\begin{lemma}[{Freed-Hopkins~\cite[\S 10]{freed_reflection_2021}}]
\label{tangstrs}
\hfill
\begin{enumerate}
    \item\label{pincpc} Pin\textsuperscript{$\tilde c+$} structures are naturally equivalent to $(B\O(2), -V)$-twisted spin structures, where $V\to B\O(2)$ is the tautological bundle.
    \item Pin\textsuperscript{$+$} structures are naturally equivalent to $(B\O(1), -\sigma)$-twisted spin structures, where $\sigma\to B\O(1)$ is the tautological bundle.
    \item\label{twspinc} Spin\textsuperscript{$c$} structures are naturally equivalent to $(B\U(1), -L)$-twisted spin structures, where $L\to B\U(1)$ is the tautological complex line bundle.
\end{enumerate}
\end{lemma}
\begin{remark}[Alternate characterizations]
Stolz~\cite[\S 6]{Sto88} showed that \pinp structures are naturally equivalent to $(B\O(1), 3\sigma)$-twisted spin structures, and it is implicit in Stong~\cite[Chapter XI]{Sto68} that \spinc structures are $(B\U(1), L)$-twisted spin structures.
\end{remark}
The pullback of $V\to B\O(2)$ along the standard inclusion $B\O(1)\to B\O(2)$ is isomorphic to $\sigma\oplus\underline\R$, which means that a \pinp structure on a vector bundle $E\to X$ induces a \pincp structure: a spin structure on $E - f^*\sigma$, where $f$ is some map $X\to B\O(1)$, is equivalent data to a spin structure on $E - f^*\sigma\oplus\underline\R$.

Similarly, the pullback of $V\to B\O(2)$ along the map $B\U(1)\to B\O(2)$ induced by the inclusion $\U(1)\cong\SO(2)\hookrightarrow\O(2)$ is $L$. Thus, analogously to the way a \pinp structure defines a \pincp structure, a \spinc structure also induces a \pincp structure. In particular, complex manifolds have canonical \pincp structures induced from their canonical \spinc structures.
\begin{definition}
It follows from \cref{tangstrs}, part~\eqref{twspinc}, that a \spinc structure on a manifold $M$ is equivalent data to a complex line bundle $L\to M$ and a spin structure on $TM\oplus L$. If $M$ is an almost complex manifold, this is equivalent to the condition $c_1(TM\oplus L)\bmod 2 = 0$, i.e.\ $c_1(M)\equiv c_1(L)\bmod 2$ by the Whitney sum formula. Since $c_1(M) = c_1(\mathrm{Det}_\C(TM))$, we can also use the determinant bundle to characterize \spinc structures.

There is a canonical isomorphism $H^2(\CP^n;\Z)\cong\Z$ sending $c_1(\mathcal O(m))\mapsto m$, and $\mathrm{Det}_\C(T\CP^n)\cong \mathcal O(-(n+1))$, so a \spinc structure on $\CP^n$ is equivalent data to an integer $m$ such that $m\equiv n+1\bmod 2$: then $TM\oplus \mathcal O(m)$ admits a spin structure, and its spin structures are a torsor over $H^1(\CP^n;\Z/2)= 0$, so this spin structure is unique. We let $(\CP^n, m)$ denote the \spinc manifold $\CP^n$ with the \spinc structure defined by $\mathcal O(m)$ in this way.
\end{definition}
Thus the \spinc structure on $\CP^n$ induced by its complex structure is $(\CP^n, -(n+1))$. If we refer to $\CP^n$ as a \spinc manifold without clarifying, we mean this structure.
\begin{lemma}[{Freed-Hopkins~\cite[Theorem 9.87]{freed_reflection_2021}}]
\label{pincp_Z23}
There is an isomorphism $\Omega_4^{\Pin^{\tilde c+}}\cong\Z_2^3$.
\end{lemma}
\begin{proposition}
\label{the_pincp_mflds}
The bordism classes of the following three manifolds are linearly independent in $\Omega_4^{\Pin^{\tilde c+}}$, and therefore form a basis.
\begin{enumerate}
    \item $\RP^4$, with \pincp structure induced from either of its two \pinp structures.
    \item $\CP^2$, with \pincp structure induced from the \spinc structure $(\CP^2, -1)$.
    \item $\CP^1\times\CP^1$, with \pincp structure induced from its complex structure.
\end{enumerate}
\end{proposition}
\begin{proof}
The fact that the bordism classes of $\RP^4$ and $\CP^1\times\CP^1$ are linearly independent in $\Omega_4^{\Pin^{\tilde c+}}$ is shown in~\cite[Proposition A.25]{DYY23}. Thus it suffices to find a bordism invariant $\xi\colon \Omega_4^{\Pin^{\tilde c+}}\to\Z_2$ which vanishes on $\RP^4$ and $\CP^1\times\CP^1$, but is nonzero on $(\CP^2, -1)$.

Given a \pincp manifold $X$, let $E\to X$ denote the rank-$2$ vector bundle associated to the \pincp structure: by \cref{tangstrs}, part~\eqref{pincpc}, a \pincp structure is a $(B\O(2), -V)$-twisted spin structure, so we have a map $f\colon X\to B\O(2)$, and $E\coloneqq f^*(V)$. By a standard argument due to Pontryagin~\cite{Pon47} (see Milnor-Stasheff~\cite[Theorem 4.9]{MS74}), $\xi\colon (X,E)\mapsto \int_X w_2(E)^2$ is a bordism invariant $\Omega_4^{\Pin^{\tilde c+}}\to\Z_2$.

If the \pincp structure on $X$ is induced from a \pinp structure, then as discussed above the pullback map of $E$ factors through $B\O(1)$ and therefore $E \cong L\oplus\underline\R$ for some real line bundle $L$. Thus in this case $w_2(E) = 0$, so $\xi(\RP^4) = 0$.

To show $\xi(\CP^1\times\CP^1) = 0$, we use that since the \pincp structure on $\CP^1\times\CP^1$ is induced from its complex structure,
\begin{equation}
\begin{aligned}
    \xi(\CP^1\times\CP^1) &= \int_{\CP^1\times\CP^1} w_2(\mathrm{Det}_{\C}(T(\CP^1\times\CP^1))\\
    &= \int_{\CP^1\times\CP^1} c_1(\mathrm{Det}_{\C}(T(\CP^1\times\CP^1))\bmod 2\\
    &= \int_{\CP^1\times\CP^1} c_1(\CP^1\times\CP^1)\bmod 2.
\end{aligned}
\end{equation}
Since $\CP^1\cong S^2$ has a spin structure, so does $\CP^1\times\CP^1$, and therefore its first Chern class is even, so $\xi(\CP^1\times\CP^1) = 0$.

For $(\CP^2, -1)$, $E = \mathcal O(-1)$, which has odd Chern class, so $w_2(E)\ne 0$. Since $H^*(\CP^2;\Z_2)\cong\Z_2[a]/(a^2)$ with $\lvert a\rvert  = 2$, then as soon as we know $w_2(E)\ne 0$ we see $w_2(E)^2$ is the unique nonzero element of $H^4(\CP^2;\Z_2)$, so $\xi(\CP^2, -1) =1$.
\end{proof}
Now that we know a set of generators, we can state the main theorem of this appendix, which is the calculation of the twisted Atiyah-Bott-Shapiro map on these generators.
\begin{theorem}\label{appendix_thm}
The twisted Atiyah-Bott-Shapiro map $\ABS_{-2}\colon \Omega_4^{\Pin^{\tilde c+}}\to\KO_2\cong\Z_2$ sends $[\RP^4]\mapsto 0$, $[\CP^2, -1]\mapsto 0$, and $[\CP^1\times\CP^1]\mapsto 1$.
\end{theorem}
The key fact that enables us to get at $\ABS_{-2}$ is:
\begin{proposition}[{Freed-Hopkins~\cite[Proposition 10.27]{freed_reflection_2021}}]
\label{FH_recast}
For $-3\le s\le -1$, the twisted Atiyah-Bott-Shapiro map $\ABS_s$ factors as
\begin{equation}
    \Omega_n^{H(s)} \xrightarrow{\sm_V} \Omega_{n+s}^\Spin(B\O(-s)) \xrightarrow{c} \Omega^\Spin_{n+s} \xrightarrow{\ABS_0} \KO_{n+s},
\end{equation}
where $V\to B\O(-s)$ is the tautological bundle, $\sm_V$ is the \term{Smith homomorphism} defined by taking a manifold representative of the Poincaré dual of the Euler class of $V$, and $c$ is the map forgetting the data of the map to $B\O(-s)$.
\end{proposition}
See Miyazawa~\cite[Theorem 3.1]{miyazawa_localization_2025} for a generalization of this theorem.
\begin{remark}
Freed-Hopkins do not define their map in exactly this way. Instead of $\sm_V$, they use the map defined by the zero section of the tautological bundle; see~\cite[Proposition 3.17]{mathSmith} for a proof identifying this with the Smith homomorphism. Instead of $\ABS_0\circ c$, they tensor $\ABS_0$ with a map to $\KO$-theory corresponding to the trivial line bundle on $B\O(-s)$, but the trivial line bundle pulls back from the point so we may use the forgetful map $c$.
\end{remark}
\begin{remark}
The Euler class mentioned in \cref{FH_recast} is not the usual Euler class, but an analogue in the spin cobordism generalized cohomology theory. This Euler class has subtle behavior and can be tricky to calculate: see~\cite[Appendix B]{mathSmith}. For this reason, we will for the most part calculate the Smith homomorphism indirectly in this section.
\end{remark}
The Smith homomorphisms in \cref{FH_recast} can be fit into long exact sequences which often can be explicitly computed. Focusing again on $s = -2$, we have:
\begin{proposition}\label{the_Smith_LES}
There is a long exact sequence
\begin{equation}\label{appendix_Smith_LES}
    \dotsb\to
    \Omega_4^{\Pin^+} \xrightarrow{i}
    \Omega_4^{\Pin^{\tilde c+}} \xrightarrow{\sm_V}
    \Omega_2^{\Spin}(B\O(2)) \xrightarrow{\delta}
    \Omega_3^{\Pin^+} \to\dotsb
\end{equation}
where $i$ is the map on bordism corresponding to the induced \pincp structure on a \pinp manifold described above and $\delta$ applied to the bordism class of a spin manifold $\Sigma$ and a rank-$2$ vector bundle $E\to \Sigma$ is the bordism class of the sphere bundle $S(E)$ with a certain \pinp structure.
\end{proposition}
\begin{proof}
Let $E\to X$ be a virtual vector bundle and $F\to X$ be a vector bundle of rank $r$. Let $p\colon S(F)\to X$ be the sphere bundle of $F$. Introduce the following three tangential structures:
\begin{enumerate}
    \item a $\xi$-structure is a $(S(F), p^*(E))$-twisted spin structure,
    \item a $\eta$-structure is an $(X, E)$-twisted spin structure, and
    \item a $\zeta$-structure is an $(X, E\oplus F)$-twisted spin structure.
\end{enumerate}
Then~\cite[Corollary 5.8]{mathSmith} there is a long exact sequence
\begin{equation}
\label{general_Smith_LES}
\dotsb\to
\Omega_n^{\xi} \xrightarrow{p_*}
\Omega_n^{\eta} \xrightarrow{\sm_F}
\Omega_{n-r}^{\zeta} \xrightarrow{\delta}
\Omega_{n-1}^{\xi} \to \dotsb,
\end{equation}
called the \term{Smith long exact sequence}, where $\sm_F$ is the Smith homomorphism associated to $F$ and $\delta$ is induced by taking the sphere bundle of the pullback of $F$ with a certain $\xi$-structure.

Let $X = B\O(2)$ and $V\to B\O(2)$ denote the tautological bundle. Then let $E = -V$ and $F = V$, so that a $\zeta$-structure is a spin structure with a map to $B\O(2)$ and, by \cref{tangstrs}, a $\eta$-structure is equivalent to a \pincp structure.

There is a homotopy equivalence $S(V)\simeq B\O(1)$ such that the bundle map $p\colon S(V)\to B\O(2)$ is identified with $i\colon B\O(1)\to B\O(2)$,\footnote{This is a standard result; one non-original reference is~\cite[Example 7.57]{mathSmith}.} so a $\xi$-structure is a $(B\O(1), -i^*(V))$-twisted spin structure; as noted above, this is equivalent to a $(B\O(1), -\sigma)$-twisted spin structure and therefore by \cref{tangstrs} a \pinp structure. This finishes the identification of this Smith long exact sequence with the one in the theorem statement.
\end{proof}
\begin{corollary}\label{pinp_vanishing}
For any closed \pinp $4$-manifold $X$, $\ABS_{-2}(X) = 0$. In particular, $\ABS_{-2}(\RP^4) = 0$.
\end{corollary}
\begin{proof}
Exactness of~\eqref{appendix_Smith_LES} implies $\sm_V\circ i = 0$, so $\sm_V(X) = 0$; \cref{FH_recast} tells us that $\ABS_{-2}$ factors through $\sm_V$, so $\ABS_{-2}(X) = 0$ too.
\end{proof}
That's one-third of \cref{appendix_thm} right there!

For $(\CP^2, -1)$ and $\CP^1\times\CP^1$ we have to perform a more detailed analysis: $\CP^2$ has no \pinp structure (as that combined with an orientation would be a spin structure, but $\CP^2$ is not spin), and though $\CP^1\times\CP^1$ has a \pinp structure, that structure does not induce the \pincp structure we use in this section.

\begin{definition}
Define maps $\varphi_1,\varphi_2,\varphi_3\colon\Omega_2^\Spin(B\O(2))\to\Z_2$ as follows on a closed spin $2$-manifold $\Sigma$ with rank-$2$ vector bundle $E\to\Sigma$.
\begin{enumerate}
    \item $\varphi_1 = \ABS_0\circ c$, as in \cref{FH_recast}.
    \item $\varphi_2$ is the composition
    \begin{equation}
        \Omega_2^\Spin(B\O(2))\xrightarrow{\det} \Omega_2^\Spin(B\O(1))\xrightarrow{\sm_\sigma}
        \Omega_1^{\Pin^-}\cong\Z_2,
    \end{equation}
    where $\det$ is induced from the determinant map $\O(2)\to\O(1)$, $\sm_\sigma$ is the Smith homomorphism introduced by Anderson-Brown-Peterson~\cite{ABP69}, which takes a manifold representative of the Poincaré dual of the Euler class of the principal $\O(1)$-bundle, and the isomorphism $\Omega_1^{\Pin^-}\cong\Z_2$ was established by (\textit{ibid.}, Theorem 5.1).
    \item $\varphi_3(\Sigma, E) = \int_\Sigma w_2(E)$.
\end{enumerate}
\end{definition}
\begin{proposition}\label{2d_spin_O2}
The following map is an isomorphism:
\begin{equation}\label{what_is_spinO2}
    \boldsymbol\varphi \coloneqq (\varphi_1,\varphi_2,\varphi_3)\colon \Omega_2^\Spin(B\O(2))\overset\cong\longrightarrow \Z_2\oplus \Z_2\oplus\Z_2.
\end{equation}
The bordism classes of the following manifolds form the basis for $\Omega_2^\Spin(B\O(2))$ dual to $(\varphi_1,\varphi_2,\varphi_3)$.
\begin{itemize}
    \item $(\Snb^1\times \Snb^1, \underline\R^2)$, where $\Snb^1$ refers to the nonbounding spin structure on the circle. $\boldsymbol\varphi(\Snb^1\times \Snb^1, \underline\R^2) = (1,0,0)$.
    \item $(\Snb^1\times S_b^1, \sigma_R\oplus\underline\R)$, where $S_b^1$ refers to the bounding spin structure on the circle and $\sigma_R\to \Snb^1\times S_b^1$\footnote{The ``R'' 
in $\sigma_R$ refers to the \textbf{R}ight-hand factor of $S^1$. } is the pullback of the Möbius bundle $\sigma\to S^1$ by the projection onto the second factor of $\Snb^1\times S_b^1$. $\varphi(\Snb^1\times S_b^1, \sigma_R\oplus\underline\R) = (0,1,0)$.
    \item $(\CP^1, \mathcal O(1))$: $\boldsymbol\varphi(\CP^1, \mathcal O(1)) = (0,0,1)$.
\end{itemize}
\end{proposition}
\begin{proof}
Mitchell-Priddy~\cite[Theorem C]{MP84} show that, modulo odd-primary torsion, for any generalized homology theory $h_*$, there is a natural map $\psi_1\colon h_*(B\O(2))\to h_*(B\SO(3))$ and an isomorphism
\begin{equation}
    (c, \psi_1, \psi_2, \det)\colon h_n(B\O(2)) \overset\cong\longrightarrow h_n(\pt) \oplus  \widetilde h_n(B\SO(3)) \oplus h_n(L(2)) \oplus \widetilde h_n(B\O(1))
\end{equation}
for a certain spectrum
$L(2)$ and map $\psi_2\colon B\O(2)\to L(2)$. Bayen~\cite[\S 3.5.3, \S 3.6.3]{Bay94} shows $\Omega_k^\Spin(L(2))$ vanishes in degrees $3$ and below, so we will not need to worry about this factor. Wan-Wang~\cite[\S 5.5.3]{WW19} show $\widetilde\Omega_2^\Spin(B\SO(3)))\cong\Z_2$, and Anderson-Brown-Peterson~\cite{ABP69} show $\widetilde\Omega_2^\Spin(B\O(1))\cong\Z_2$. The additional hypothesis on odd-primary torsion can be removed: Randal-Williams~\cite[\S 5.1]{RW08} shows that the odd-primary torsion in $\widetilde h_*(B\O(2))$ coincides with that of a $3$-connected space, meaning there can be none in degree $2$. Thus we have the abstract isomorphism~\eqref{what_is_spinO2} and the fact that $\varphi_1$ and $\varphi_2$ are linearly independent, but we still need to address $\varphi_3$. The integral of a stable characteristic class is a bordism invariant by an argument of Pontryagin~\cite{Pon47} (see also Milnor-Stasheff~\cite[Theorem 4.9]{MS74}), so $\varphi_3$ indeed defines a map $\Omega_2^\Spin(B\O(2))\to\Z_2$; we need to show this map is linearly independent of $\varphi_1$ and $\varphi_2$. To do so, we will calculate $\boldsymbol\varphi$ on the three surfaces in the theorem statement.
\begin{itemize}
    \item $\varphi_2$ and $\varphi_3$ by definition vanish on trivial bundles, so $\boldsymbol\varphi(\Snb^1\times\Snb^1, \underline\R^2) = (?, 0, 0)$; for the value of $\varphi_1$ observe that $\ABS_0\circ c$ is the Arf invariant, which equals $1$ on $\Snb^1\times\Snb^1$.
    \item For $(\Snb^1\times S_b^1, \sigma_R\oplus\underline\R)$, we have $\Snb^1\times S_b^1 = \partial(\Snb^1\times D^2)$, so $c$ kills this manifold and therefore $\varphi_1$ does too. For $\varphi_2$, $\mathrm{Det}(\sigma_R\oplus\underline\R)\cong\sigma_R$. This bundle is trivializable when restricted to $\Snb^1\times\set x\subset \Snb^1\times S_b^1$ for any $x\in S_b^1$, which means $\Snb^1$ is Poincaré dual to the Euler class of $\sigma_R$ and therefore $\sm_\sigma(\Snb^1\times S_b^1, \sigma_R) = \Snb^1$, whose class in $\Omega_1^{\Pin^-}$ is nonzero~\cite[Theorem 2.1]{KT90}. Thus $\varphi_2(\Snb^1\times S_b^1,\sigma_R\oplus\underline\R) = 1$.
    For $\varphi_3$, $w_2(\sigma_R\oplus\underline\R) = w_2(\sigma_R) = 0$, because $\sigma_R$ is a line bundle.
    \item Finally, $(\CP^1, \mathcal O(1))$: since $\CP^1\cong S^2$ is simply connected, it has a unique spin structure, which bounds $D^3$ and therefore has trivial Arf invariant, so $\varphi_1(\CP^1, \mathcal O(1)) = 0$. Since $\mathcal O(1))$ is complex, it is oriented, so its real determinant bundle vanishes, and therefore $\varphi_2(\CP^1, \mathcal O(1)) = 0$. Since $c_1(\mathcal O(1))\mapsto 1$ under the isomorphism $H^2(\CP^1;\Z)\xrightarrow{\cong}\Z$ defined by the orientation induced by the complex structure, $w_2(\mathcal O(1)) = c_1(\mathcal O(1))\bmod 2$ is the nonzero element of $H^2(\CP^1;\Z_2)\cong\Z_2$, and therefore $\int_{\CP^1} w_2(\mathcal O(1)) = 1$.
\end{itemize}
Thus we have shown that the bordism classes of these three surfaces are linearly independent, and dual to the three invariants in $\boldsymbol\varphi$, as promised.
\end{proof}
Recall the map $\delta$ from \cref{the_Smith_LES}.
\begin{lemma}\label{delta_is_phi2}
There is a (necessarily unique) isomorphism $q\colon \Omega_3^{\Pin^+}\xrightarrow{\cong}\Z_2$, and the composition $q\circ \delta\colon \Omega_2^{\Spin}(B\O(2))\to\Z_2$ equals $\varphi_2$.
\end{lemma}
\begin{proof}
The calculation $\Omega_3^{\Pin^+}\cong\Z_2$ is due to Giambalvo~\cite[\S 2]{Gia73}. To identify $q\circ\delta = \varphi_2$, it suffices by \cref{2d_spin_O2} to show $\delta(\Snb^1\times\Snb^1, \underline\R^2) = 0$, $\delta(\Snb^1\times S_b^1, \sigma_R\oplus\underline\R) = 1$, and $\delta(\CP^1, \mathcal O(1)) = 0$.

First observe that $\sm_V$ is not surjective: its domain and codomain are both sets of size $8$ (\cref{pincp_Z23}, resp.\ \cref{2d_spin_O2}) but $\sm_V(\RP^4) = 0$ (\cref{pinp_vanishing}). Since $\sm_V$ is not surjective, exactness of~\eqref{appendix_Smith_LES} implies $\delta\ne 0$. Therefore if we can show $\delta(\Snb^1\times\Snb^1, \underline\R^2) = 0$ and $\delta(\CP^1, \mathcal O(1)) = 0$, then it must follow that $\delta(\Snb^1\times S_b^1, \sigma_R\oplus\underline\R) = 1$.

For $\Snb^1\times\Snb^1$, the vector bundle is trivial, so the total space of its sphere bundle is $T^3$ with some \pinp structure. Since $T^3$ is orientable, this \pinp structure is induced from a spin structure, but $\Omega_3^\Spin = 0$~\cite{Mil63} so $T^3$ bounds some compact spin $4$-manifold $X$. This is also a \pinp null-bordism of $T^3$, so $\delta(\Snb^1\times\Snb^1, \underline\R^2) = 0$.

For $\CP^1$, the sphere bundle map $S(\mathcal O(1))\to\CP^1$ is one definition of the Hopf fibration, so the total space is $S^3$. The argument in the previous paragraph shows that any \pinp structure on any closed, orientable $3$-manifold is null-bordant, so $\delta(\CP^1, \mathcal O(1)) = 0$.
\end{proof}
Thus, by exactness of~\eqref{appendix_Smith_LES}, $\varphi_2\circ\sm_V = 0$.
\begin{lemma}\label{app_3}
$\boldsymbol\varphi\circ\sm_V(\CP^1\times\CP^1) = (1,0,0)$.
\end{lemma}
\begin{proof}
By \cref{delta_is_phi2}, $\varphi_2\circ\sm_V = 0$.

To compute $\varphi_3\circ\sm_V(\CP^1\times\CP^1)$, let $i\colon \Sigma\hookrightarrow\CP^1\times\CP^1$ be a manifold representative for the Poincaré dual of the Euler class of the principal $\O(2)$-bundle $V\to \CP^1\times\CP^1$ associated to the \pincp structure. If $\nu\to\Sigma$ denotes the normal bundle to the embedding $i$, then there is an isomorphism $i^*(V)\cong\nu$ and $\sm_V(\CP^1\times\CP^1)$ is by definition the class of $(\Sigma, i^*(V))$ in $\Omega_2^\Spin(B\O(2))$. Apply the Whitney sum formula to the decomposition $T\Sigma\oplus\nu\cong i^*(T(\CP^1\times\CP^1))$ to deduce
\begin{subequations}
\begin{align}
    \label{the_w1_part}
    w_1(\Sigma) + w_1(\nu) &= i^*(w_1(\CP^1\times\CP^1))\\
    \label{the_w2_part}
    w_2(\Sigma) + w_1(\Sigma)w_1(\nu) + w_2(\nu) &= i^*(w_2(\CP^1\times\CP^1)).
\end{align}
\end{subequations}
Since $\CP^1\cong S^2$, it has a spin structure, so $\CP^1\times\CP^1$ does as well, and therefore $w_i(\CP^1\times\CP^1) = 0$ for $i = 1,2$. Thus~\eqref{the_w1_part} simplifies to $w_1(\Sigma) = w_1(\nu)$, and so~\eqref{the_w2_part} simplifies to $w_2(\Sigma) + w_1(\Sigma)^2 + w_2(\nu) = 0$. The Wu formula implies $w_2(\Sigma) + w_1(\Sigma)^2 = 0$ because the Wu class $v_2 = w_2 + w_1^2$ vanishes on closed $2$-manifolds such as $\Sigma$, so we have calculated that $w_2(\nu) = 0$ and therefore
\begin{equation}
    \varphi_3(\CP^1\times\CP^1) \coloneqq \int_\Sigma w_2(\nu) = 0.
\end{equation}
Since $\varphi_i\circ\sm_V(\CP^1\times\CP^1) = 0$ for $i = 2,3$, to show $\varphi_1\circ\sm_V(\CP^1\times\CP^1) = 1$ is the same as showing $\boldsymbol\varphi(\CP^1\times\CP^1)\ne 0$. Since $\boldsymbol\varphi$ is an isomorphism (\cref{2d_spin_O2}), this is the same as showing $\sm_V(\CP^1\times\CP^1)\ne 0$, which by exactness of~\eqref{appendix_Smith_LES} is equivalent to showing $[\CP^1\times\CP^1]\not\in\Im(i)$. The domain of $i$, $\Omega_4^{\Pin^+}$, is a cyclic group~\cite[\S 2]{Gia73} and by construction $[\RP^4]\in\Im(i)$, so since the bordism classes of $\RP^4$ and $\CP^1\times\CP^1$ are linearly independent (\cref{the_pincp_mflds}), $[\CP^1\times\CP^1]$ cannot also be in $\Im(i)$. Thus $\varphi_1\circ\sm_V(\CP^1\times\CP^1) = 1$.
\end{proof}
\begin{lemma}\label{app_2}
$\varphi_1\circ\sm_V(\CP^2, -1) = 0$.
\end{lemma}
\begin{proof}
One characterization of $\sm_V(\CP^2, -1)$ is that it is the bordism class of the zero set of any section of $V\to\CP^2$ which is transverse to the zero section~\cite[Definition 3.7]{mathSmith}. For the \pincp structure $(\CP^2, -1)$, $V = \mathcal O(-1)$, and the zero set of any such section is isotopic to the standard embedding $\CP^1\to\CP^2$. Thus $\sm_V(\CP^2, -1)\in\Omega_2^\Spin(B\O(2))$ is the bordism class of $\CP^1$ with some spin structure and some rank-$2$ vector bundle. The map $\varphi_1$ forgets the vector bundle, so $\varphi_1\circ\sm_V(\CP^2, -1)\in\Omega_2^\Spin$ is the bordism class of $\CP^1\cong S^2$ with some spin structure. Since $S^2$ is simply connected, it has a unique spin structure, which therefore is the spin structure appearing at the boundary $S^2\cong\partial D^3$, where $D^3$ is given its canonical (also unique) spin structure. Therefore $[\CP^1] = 0$ in $\Omega_2^\Spin$ and therefore $\varphi_1\circ\sm_V(\CP^2, -1)= 0$.
\end{proof}
\begin{remark}
It is possible to show $\boldsymbol\varphi\circ\sm_V(\CP^2, -1) = (0,0,1)$ similarly to the proof of \cref{app_3}.
\end{remark}
Since the first component of $\boldsymbol\varphi\circ\sm_V$ is $\ABS_{-2}$, \cref{pinp_vanishing,app_3,app_2} finish the proof of \cref{appendix_thm}.

\bibliographystyle{alpha}
\bibliography{CK_Zotero}

\end{document}